\documentclass[journal]{IEEEtran}

\usepackage{multicol, blindtext}
\usepackage{tabularx}
\usepackage{amsmath}
\usepackage{amsmath, nccmath}

\usepackage{amsthm}

\usepackage[utf8]{inputenc}
\usepackage[T1]{fontenc}
\usepackage{mathtools}
\usepackage{float}
\usepackage{graphicx}
\usepackage{mathtools}
\usepackage{amssymb}
\usepackage{booktabs}
\usepackage{breqn}
\usepackage{graphicx}
\usepackage[caption=false]{subfig}
\usepackage{lipsum}
\usepackage{subfig}
\usepackage{graphicx}
\DeclarePairedDelimiter\ceil{\lceil}{\rceil}
\DeclarePairedDelimiter\floor{\lfloor}{\rfloor}
\usepackage{amsmath}
\usepackage{algorithm,algorithmic}
\usepackage{flushend}
\usepackage[sorting=none]{biblatex}
\bibliography{references}
\flushend

\newtheorem{corollary}{Corollary}
\newtheorem{lemma}{Lemma}

\begin{document}

\title{Trimmed Minimum Error Entropy for Robust Online Regression}

\author{Sajjad~Bahrami,~\IEEEmembership{Student Member,~IEEE,}
        and~Ertem~Tuncel,~\IEEEmembership{Senior~Member,~IEEE,}
        \\ {\tt\small sbahr003@ucr.edu}  {\tt\small and ertem@ece.ucr.edu}
\thanks{ The material of this paper was presented in part at the 2020 IEEE
International Symposium on Information Theory.}}

\maketitle

\begin{abstract}
In this paper, online linear regression in environments corrupted by non-Gaussian noise (especially heavy-tailed noise) is addressed. In such environments, the error between the system output and the label also does not follow a Gaussian distribution and there might exist abnormally large error samples (or outliers) which mislead the learning process. The main challenge is how to keep the supervised learning problem least affected by these unwanted and misleading outliers. In recent years, an information theoretic algorithm based on Renyi's entropy, called minimum error entropy (MEE), has been employed to take on this issue. However, this minimization might not result in a desired estimator inasmuch as entropy is shift-invariant, i.e., by minimizing the error entropy, error samples may not be necessarily concentrated around zero. In this paper, a quantization technique is proposed by which not only aforementioned need of setting errors around the origin in MEE is addressed, but also major outliers  are rejected from MEE-based learning and MEE performance is improved from convergence rate, steady state misalignment, and testing error points of view. 
\end{abstract}
\section{Introduction}

\IEEEPARstart{I}{n} many real-world signal processing and machine learning applications, especially nonlinear topologies, we encounter non-Gaussian probability density functions (PDFs). In such realistic scenarios, the noise might have a heavy-tailed distribution, there might exist severe outlier noise, and the error PDF might even change in time \cite{1}. As a real-world example, we can point at underwater communications in which the Gaussian assumption cannot be made anymore due to the existence of impulsive noise \cite{2,3,4,5}.

Although even up to now, the most commonly used cost functions are based on the moments of the data \cite{6}, e.g., variance, skewness, and kurtosis, which are the 2nd, 3rd, and 4th central moments, respectively, when the error is non-Gaussian, these cost functions are not reliable and we need to take into account higher order statistics of data as well \cite{1}. Therefore, we need to look for other general and robust descriptors of the data statistics that improve the algorithm performance.

In recent years, two effective cost functions, namely entropy and correntropy, have been employed by information theoretic learning in non-Gaussian environments as superior alternatives to the famous and most commonly used cost function, i.e., mean square error (MSE) \cite{7,8}. The relation between algorithms based on entropy and correntropy has been investigated in \cite{8.5}. In the literature, entropy and correntropy are sometimes interpreted as counterparts of variance and correlation, respectively \cite{1}. Both entropy and correntropy involve higher-order data statistics and therefore they are expected to outperform MSE, which only contains second-order moment.

Specifically, entropy as a robust information theoretic cost function contains all higher-order moments and although the algorithm based on error entropy is computationally more expensive than error correntropy, entropy is a more general descriptor of the underlying error statistics \cite{10}. 
Moreover, previous work have tackled computational bottleneck of entropy approximation in large-scale data sets, e.g., \cite{12} and \cite{10}, where fast Gauss transform and quantization were employed, respectively, to reduce computational complexity. Some results regarding consistency, robustness, uniqueness of the solution, sufficient and necessary conditions for MEE algorithm can be found in \cite{12.1} and \cite{12.2}. Moreover, authors in \cite{12.2.1} show that even when large outliers exist in both input and output variables, MEE can result in a very close solution to the optimum value. Some applications of entropy minimization in adaptive system training, neural networks, blind deconvolution, parameter estimation, blind source separation, digital communication channel equalization, and channel estimation for massive multiple input multiple output (MIMO) communication can be found in \cite{7,13}, \cite{14}, \cite{15}, \cite{16}, \cite{17}, \cite{18} and  \cite{18.1}. Bayesian estimation based on MEE is also addressed in \cite{12.3} and \cite{12.4}.

Outlier effect is one of the main challenges that we must deal with when we work with non-Gaussian (multimodal, impulsive, heavy-tailed, etc.) environments. Although an outlier may have either a negative or positive role \cite{19}, outliers in our problem are not informative and arose from non-Gaussian measurement noise. Throughout the paper, we mitigate this outlier effect in a non-Gaussian environment. We consider an online linear regression problem in which we receive new data samples at each time instant and use it to update the parameters of the underlying system. Adaptive noise cancellation, system identification, channel estimation, etc. \cite{19.1} are some examples of online linear regression applications. Our goal is to minimize the error between linear system output and labels even when environment is severely affected by outliers. To this end, we employ entropy. As entropy denotes the average dispersion of data, we minimize it to concentrate the errors \cite{20}. In other words, when we use error entropy as cost function and minimize it we indeed attempt to ideally set the distribution of error as an impulse. However, as error entropy minimization is shift-invariant \cite{12.2,12.3}, we must take some further steps to concentrate errors specifically around $e=0$.

In this paper, we employ a nonlinear quantization technique to detect major outliers in error samples. Then, we mitigate their destructive effect in the learning by not considering them in both processes of MEE-based learning and concentrating error samples around $e=0$. This new algorithm results in faster convergence to a lower steady state misalignment and achieves a smaller testing error. 

The rest of the paper is organized as follows. In Section II, first we review MEE and maximum correntropy criterion (MCC) and compare them, then we discuss different methods for concentrating error samples around the origin in MEE. Section III is devoted to proposing our outlier detection technique based on nonlinear quantization and deploying it into the MEE to improve its performance. In section IV, simulation results are presented. Ultimately, we conclude the paper in section V.   

\textit{Notation:} Random variables and their realization are shown by uppercase and lowercase letters, respectively while vectors are denoted by boldface letters. Moreover, $\textbf{E}\{.\}$ stands for expectation operator and $\parallel .\parallel $ denotes 2 norm. 
\begin{figure}
\centering
\includegraphics[scale=.53]{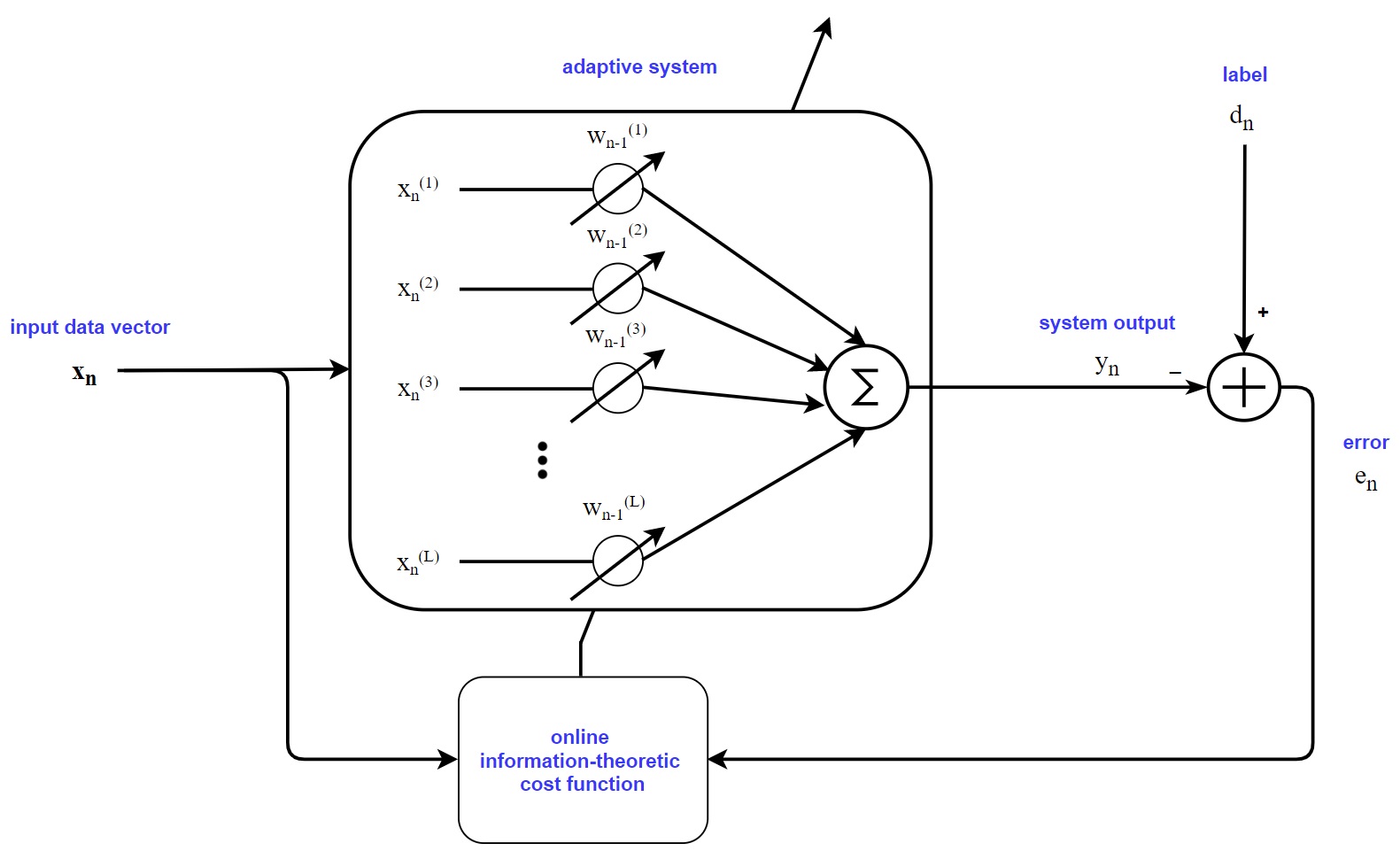}
\caption{Online linear regression based on information-theoretic cost functions.}
\label{fig1}
\end{figure}

\section{MEE vs. MCC}
We start with restating the online linear regression (or linear adaptive filtering) problem illustrated in Figure \ref{fig1}. The goal is to find the parameters of a linear system $y_n=\textbf{x}_n^T\textbf{w}_{n-1}$ in which $\textbf{x}_n=\big[ x_n^{(1)},x_n^{(2)},\cdots ,x_n^{(L)}\big] ^T$ is the input vector at time instant $n$, $\textbf{w}_{n-1}=\big[ w_{n-1}^{(1)},w_{n-1}^{(2)},\cdots ,w_{n-1}^{(L)}\big] ^T$ is a vector denoting system parameters estimated at time instant $n-1$, $L$ indicates the length of the parameter vector and $y_n$ denotes system output estimated at time instant $n$. Error samples at time $n$ are obtained as $e_i=d_i-\textbf{x}_i^T\textbf{w}_{n-1}$ where $d_i$ denotes label at time instant $i$. We want to concentrate error samples around $e=0$ with time. 

Both MEE and MCC are deployed in information-theoretic learning context as robust criteria to deal with non-Gaussianity of the environment which becomes problematic if conventional MSE is used. In other words, superiority of information-theoretic cost functions becomes more clear when the error distribution is non-Gaussian and this happens if the filter topology is nonlinear (which is not of interest in this paper) or the label (or noise in the label) is non-Gaussian (which we do address in this paper). In this paper, we specifically focus on a class of non-Gaussian distributions whose tails are heavier than an exponential. Sampling from such distributions results in mostly "normal" values with a few "abnormal" values (outliers). We can name  power-law, lognormal, Weibull, Cauchy, $\alpha $-stable and mixture of Gaussians as examples of such distributions. In order to deal with error with a heavy-tailed distribution we need to incorporate higher order moments of the error (which are usually large \cite{20.1}), while conventional MSE (i.e., $ J_{MSE}(E)=\textbf{E}\{E^2\} $) contains only second order moment of the error which results in small misalignment between exact solution and estimate when the noise is Gaussian. \cite{1}.

In the rest of this section, we first provide an overview of MCC and MEE, then we show how they contain higher order statistics of the error (and consequently how they are intuitively superior to MSE) and also we provide a brief comparison of them. Finally, we challenge existing solutions for concentrating error samples around $e=0$ in MEE.   \vspace{-.5cm}

\subsection{Overview of MCC and MEE}
For both MCC and MEE we need error PDF to calculate the information-theoretic costs while this PDF is unknown. Parzen window technique is used as a non-parametric method for error PDF estimation at time instant $n$ as follows \cite{20.2,23}: 
\begin{align}
\label{parzen}
p_E(e)\approx \frac{1}{N}\sum _{i=0}^{N-1}G_{\sigma}(e-e_{n-i})=\hat{p}_E^{(n)}(e)
\end{align}
in which $N$ denotes number of error samples used for this non-parametric PDF estimation and $G_{\sigma }(.)$ is the following Gaussian kernel with kernel bandwidth $\sigma $:
\begin{align}
G_{\sigma}\left( e\right) =\frac{1}{\sqrt{2\pi}\sigma}\exp(-\frac{\parallel e\parallel ^2}{2\sigma ^2}). \nonumber 
\end{align}
In the following we see how MCC is related to error PDF by (\ref{parzen}) while MEE directly uses (\ref{parzen}) for entropy estimation.
\subsubsection{Maximum Correntropy Criterion (MCC)}
Correntropy is defined as follows:
\begin{align}
\label{correntropy}
v\left( D,Y\right) =\textbf{E}\left \{ G_{\sigma}\left( D-Y\right) \right \}= \textbf{E}\left \{ G_{\sigma}\left( E\right) \right \}=v(E), 
\end{align}
in which $D$ and $Y$ are two random variables denoting label and system output, respectively, random variable $E=D-Y$ denotes the error, and $G_{\sigma }(.)$ is the same Gaussian kernel used in (\ref{parzen}) for error PDF estimation. Sample mean approximation of correntropy (\ref{correntropy}) at time instant $n$ from data samples $\left \{ d_{n-i},y_{n-i}\right \} (\mathrm{or~equivalently} \{ e_{n-i}\}), ~i=0,\cdots ,N-1,$ is as follows: 
\begin{align}
\label{corest}
\hat{v}^{(n)}\left( D,Y\right) =\frac{1}{N}\sum _{i=0}^{N-1}G_{\sigma}\left( d_{n-i}-y_{n-i}\right) =\frac{1}{N}\sum _{i=0}^{N-1}G_{\sigma}\left( e_{n-i}\right) . 
\end{align}
The above relation is indeed the online MCC cost function optimized using stochastic gradient ascent with a computational complexity of $O(N)$ in each iteration.

Obviously,  $\hat{v}^{(n)}\left( D,Y\right) =\hat{p}_E^{(n)}\left(0\right) $ is concluded from (\ref{parzen}) and (\ref{corest}) which means that if we maximize the estimate of correntropy we indeed maximize the estimate of error PDF at $0$. This is the reason that maximum correntropy criterion is meaningful. In addition, if we take a closer look at (\ref{corest}) we can see that correntropy is a similarity measure between two random variables $D$ and $Y$. In fact, abnormally large error samples (or outliers) are given small weights and are filtered out by Gaussian kernel while ones with smaller values have larger contribution in the learning process inasmuch as they are assigned larger weights.

As stated earlier, in non-Gaussian environments we need higher-order statistics \cite{1,20.1} and correntropy provides us with that. Using Taylor expansion of Gaussian kernel in (\ref{correntropy}), we can observe that correntropy contains other even-order moments of the error PDF (if they exist) as well:
\begin{align}
v\left( D,Y\right) =\frac{1}{\sqrt{2\pi}\sigma}\sum _{k=0}^{\infty}\frac{(-1)^k}{2^k\sigma ^{2k}k!}\textbf{E}\left \{ \left( D-Y\right)^{2k}\right \}. \nonumber
\end{align}
Note that we could use other kernels for MCC \cite{23.1,23.2} and even incorporate more higher-order moments.
\subsubsection{Minimum Error Entropy (MEE)}
Since this paper presents a modification to MEE, this subsection reviews MEE in more detail. We use Renyi's entropy in the sequel which is a parametric family of entropies (Shannon entropy is its limiting case) \cite{20}. For simplicity, we use Renyi's quadratic entropy whose estimation from samples of the underlying random variable has been well studied \cite{1}. Renyi's quadratic entropy is defined as follows:
\begin{align}
H_2(E)=-\log I_2(E), \nonumber 
\end{align}
where $I_2(E)=\textbf{\mathrm{E}} \{p_E(E)\} =\int p_E^2(e)de$ is called information potential and $p_E(e)$ denotes the error PDF. Since $\log (.)$ is a monotonically increasing function, the following optimization problems are the same:
\begin{align} 
\min _{\textbf{w}}H_2(E)=\max _{\textbf{w}}I_2(E),\nonumber 
\end{align}
and hence it suffices to maximize information potential. We need to estimate the information potential. However, similar to the case of MCC, we are dealing with an extremely large data-set where we are receiving continuously new data (error) samples with time, and so it is not efficient to use a batch estimator to incorporate all data samples for information potential estimation. Therefore, we utilize an online approach and estimate information potential at time instant $n$ from past $N$ error samples as follows:
\begin{align}
\label{L1}
I_2(E)=\mathrm{\textbf{E}} \{p_E(E)\}\approx \frac{1}{N}\sum _{i=0}^{N-1}p_E(e_{n-i})=\hat{I}_2^{(n)}(E). 
\end{align}
Note that although online approaches cannot usually optimize our cost function precisely, they are able to quickly process an extremely large data-set and get close enough to the optimum solution \cite{22.1}.

As stated earlier, we do not know the error statistics, therefore we estimate $p_E(e_{n-i})$ in (\ref{L1}) from error samples by using the  Parzen window technique (\ref{parzen}) as follows:
\begin{align}
\label{L2}
p_E(e_{n-i})\approx \frac{1}{N}\sum _{j=0}^{N-1}G_{\sigma}(e_{n-i}-e_{n-j}). 
\end{align}
Substituting (\ref{L2}) into (\ref{L1}), we have the following estimate of information potential from past $N$ error samples at time instant $n$:
\begin{align}
\label{L3}
\hat{I}_2^{(n)}(E)=\frac{1}{N^2}\sum _{i=0}^{N-1}\sum _{j=0}^{N-1}G_{\sigma}(e_{n-i}-e_{n-j}). 
\end{align}
Relation (\ref{L3}) is indeed the online cost function of MEE that we optimize using stochastic gradient ascent. MEE has a computational complexity of $O(N^2)$ in each iteration.

As for the higher-order statistics, MEE contains all higher-order moments of the error PDF (if they exist) regardless of the type of the kernel we use for Parzen PDF estimation as opposed to MCC, hence we can consider it as a global descriptor of the error PDF. This is shown by using Taylor expansion of error PDF $p_E(e)$ as follows:
\begin{align}
&p_E(e)=p_E(0)+p_E^{(1)}(0)e+\frac{p_E^{(2)}(0)}{2!}e^2+\frac{p_E^{(3)}(0)}{3!}e^3+\cdots \nonumber \\
\Rightarrow &I_2(E)=\textbf{E}\{p_E(E)\} = p_E(0)+p_E^{(1)}(0)\textbf{E}\{E\} \nonumber \\
&~~~~~~~~~~~~~+\frac{p_E^{(2)}(0)}{2!}\textbf{E}\{ E^2\}+\frac{p_E^{(3)}(0)}{3!}\textbf{E}\{ E^3\}+\cdots 
\nonumber 
\end{align}
where $p_E^{(i)}(0)$ is the $ith$ derivative of PDF at $e=0$.
Note that PDFs in practice are usually smooth and consequently continuously differentiable.
\subsubsection{Comparison of MEE and MCC}
Although both MEE and MCC take into account information content of error and its higher-order statistics, MEE is expected to have a superior performance compared to MCC in general at the cost of higher computational complexity. This has been shown in many experimental results \cite{8.5,12.4,34NEW}. The problem of MCC arises from being a local criterion that takes into account mostly the errors within the Gaussian kernel bandwidth, while error modes might in fact be far from the origin. On the other hand, MEE's superior performance emerges from self-adjusting the weights of different error samples based on the error distribution itself. Therefore, MCC may not perform as efficiently as MEE in non-Gaussian noises with a light-tail or multimodal distribution \cite{1}.

As mentioned in \cite{8.5}, the difference between cost functions of MCC and MEE can be obtained by using Euclidean distance between error PDF $p_E(e)$ and Gaussian kernel as follows:
\begin{align}
\label{MCCvsMEE}
&D_{ED}\left( p_E(e),G_{\sigma }(e)\right) = \int \left( p_E(e)-G_{\sigma }(e)\right)^2de \nonumber \\
&=\overbrace{ \int \left( p_E(e)\right)^2de }^{I_2(E)}+\overbrace{\int \left( G_{\sigma }(e)\right)^2de}^{\frac{1}{2\sigma \sqrt{\pi}}} -2 \overbrace{\int p_E(e)G_{\sigma }(e)de}^{v(E)} \nonumber \\
&\Longrightarrow ~I_2(E) + \frac{1}{2\sigma \sqrt{\pi}} = 2v(E)+D_{ED}\left( p_E(e),G_{\sigma }(e)\right). 
\end{align}
Obviously, there is a difference between MCC and MEE cost functions based on (\ref{MCCvsMEE}), and consequently their optimum solution may be different as well.

\begin{corollary}
If Gaussian kernel bandwidth $\sigma $ is fixed (not adaptive) and also Euclidean distance between error PDF and Gaussian kernel is zero, then MCC and MEE are equivalent and result in the same solution.
\end{corollary} 
\begin{proof}
Easily seen from (\ref{MCCvsMEE}).
\end{proof} 
\textbf{Remark 1:} Gaussian kernel bandwidth $\sigma$ is a free parameter in both MCC and MEE cost functions that can be optimized during the learning process to increase algorithm efficiency. It indeed determines the magnitude of the weights assigned to each error sample and it is a function of error. Optimizing this bandwidth has been widely addressed in previous work, for instance by minimizing Kullback–Leibler divergence between the true and estimated error distribution, using shape of error distribution measured by its kurtosis, using instantaneous error in each iteration, changing the Gaussian kernel, using hybrid methods and so forth \cite{23.4,23.5,23.6,23.7,23.92,23.8,23.9,23.91}.
\subsection{Concentrating error samples around $e=0$ in MEE}
So far, we have discussed our expectation of MEE superiority over MCC. However, alongside the higher computational complexity of MEE compared to MCC that has been addressed in previous work as stated earlier, another difficulty associated with MEE is that the error PDF needs to be move to the origin, as entropy is shift-invariant. 
Towards that end, te following approaches have been proposed in the literature:
\begin{enumerate}
\item Adding sample mean of the labels (sample mean up to time instant $n$) to the output of the linear system as a bias term \cite{1}: Although this approach is very simple and works for labels with symmetric PDFs, environments in many real world scenarios are corrupted by asymmetric and heavy-tailed noises which contain many large outliers and consequently sample mean may be very misleading. In other words, sample mean may fail to converge in probability to the expected value, and law of large numbers does not hold in such environments \cite{21}.
\item Minimization of Error Entropy with Fiducial points (MEEF) \cite{22}: This approach suggests to consider a fiducial zero vector of arbitrary length $M$ whose elements are indeed points of reference and help to fix the peak of the error PDF at the origin. Consequently, error entropy minimization forces the PDF to approach an impulse around $e=0$. Now, information potential at each time instant $n$ using past $N$ error samples plus $M$ fiducial points is denoted by $\hat{I}_{2,F}^{(n)}(E)$ and is estimated as follows:\begin{align}
\hat{I}_{2,F}^{(n)}(E)=&\frac{1}{{(N+M)}^2}\sum _{i=0}^{N+M-1}\sum _{j=0}^{N+M-1}G_{\sigma }(e_{n-i}-e_{n-j}), \nonumber 
\end{align}where $[e_{n-N},e_{n-(N+1)},\cdots ,e_{n-(N+M-1)}]=\underline{\textbf{0}}_{1\times M}$ denotes fiducial zero vector. Then, above relation can be rewritten as follows:
\begin{align}
\label{MEEF}
&\hat{I}_{2,F}^{(n)}(E)=\frac{1}{{(N+M)}^2}\sum _{i=0}^{N-1}\sum _{j=0}^{N-1}G_{\sigma }(e_{n-i}-e_{n-j}) \nonumber \\
&+\frac{2M}{(N+M)^2}\sum _{i=0}^{N-1}G_{\sigma }(e_{n-i}) +\frac{M^2}{(N+M)^2}G_{\sigma }(0). 
\end{align}Relation (\ref{MEEF}) can be interpreted as a weighted combination of the error entropy criterion (first term on the right hand side) and the error correntropy criterion (second term on the right hand side). The first term strives to make the error PDF as close as possible to an impulse, while the second term pushes the peak of the error PDF towards $e=0$. This approach may outperform the previous one \cite{22,23.93}. However, there is an obvious trade-off: As we increase the number of fiducial points, i.e., $M$, the cost function (\ref{MEEF}) gets closer to that of MCC. Although this will make the role of the correntropy-related term more emphasized, thereby moving the peak of the error PDF towards the origin more aggressively, the accuracy of the entropy estimation would suffer. In other words, as seen in (\ref{MCCvsMEE}), this can be problematic and deteriorate the performance of MEE when the difference between MCC and MEE cost functions is not negligible.
\end{enumerate}

Due to the drawbacks of above methods for concentrating error samples around $e=0$ in MEE, we propose a new approach called \emph{Trimmed MEE}, which we describe in detail in the next section. 
\section{Proposed Trimmed MEE}
In this section, we modify MEE in order to improve its performance and also to overcome aforementioned shortcomings of the existing methods for locating error samples at the origin. The key idea behind the proposed method is to stop incorporating abnormally large errors (or major outliers) into the learning process. Recall that major outliers in error samples can  be very misleading as they differ significantly from other observations, therefore we strive to exclude them from other error samples.

The important question then becomes "how can we determine whether an error sample is a major outlier?" We use running quartiles of the error samples to detect major outliers. More specifically, as quantiles of a data set are robust quantities of data against outliers (and hence do not change significantly by a new major outlier) we use the concept of outer fences to determine two boundaries for major outlier rejection \cite{191,19}. 
\begin{figure}
\centering
\includegraphics[scale=.3]{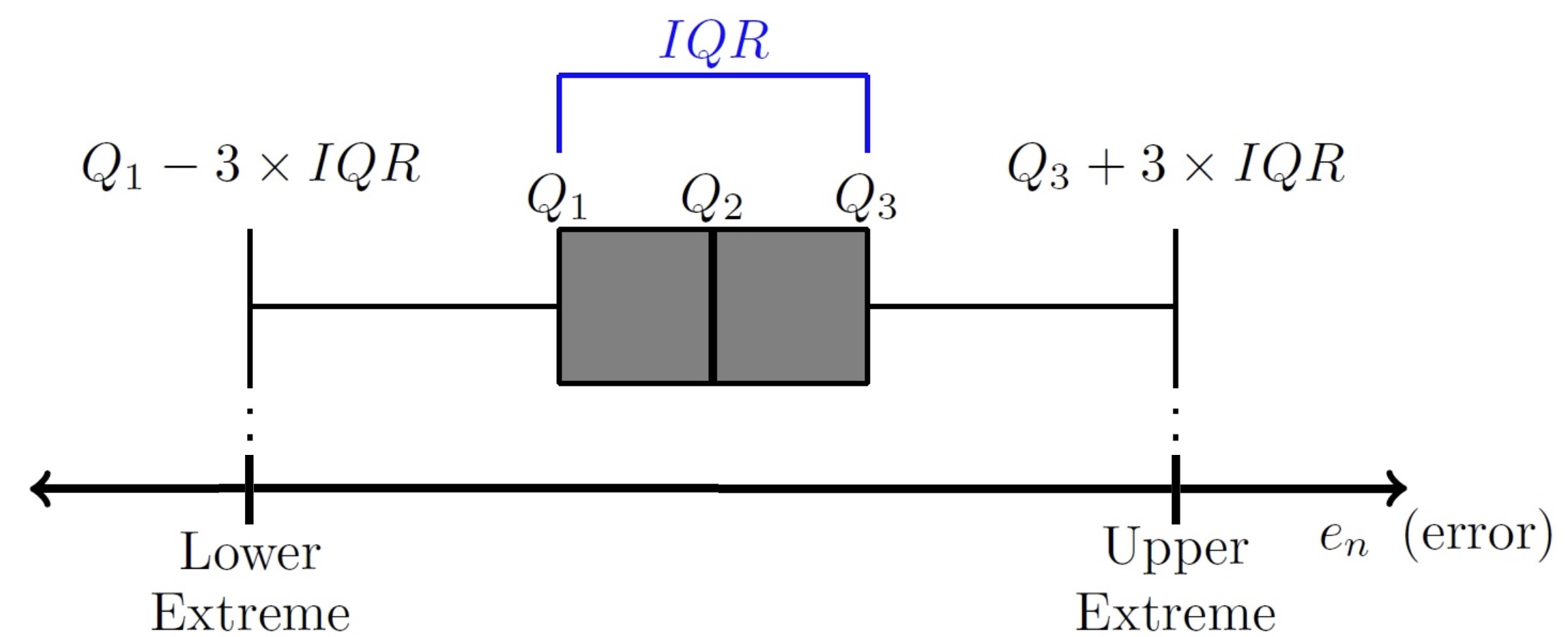}
\caption{Boundaries beyond which abnormally large error samples (major outliers) reside.}
\label{fig2}
\end{figure}
These boundaries are obtained as illustrated in Figure \ref{fig2}. In this Figure, $Q_1,~Q_2~\mathrm{and}~Q_3$ are lower quartile (or 25th percentile), median, and upper quartile (or 75th percentile), respectively, and $IQR=Q_3-Q_1$ stands for inter-quartile range. Boundaries (or outer fences) are defined as follows:
\begin{align}
\mathrm{Lower~Extreme} &=Q_1-3\times IQR, \nonumber \\ \mathrm{Upper~Extreme}&=Q_3+3\times IQR. \nonumber 
\end{align}

Obviously, in this outlier detection method, we only need to obtain $Q_1$ and $Q_3$ in each adaptation step. On the other hand, this could be challenging. For example, if we simply use order statistics to obtain these running quartiles where we simply sort all observed error samples at each time instant $n$, we would suffer from a complexity of $O\left( n\right)$ and we would have to store all previous error samples.

In the next subsection, we present our running quartile estimation technique for estimation of $Q_1$ and $Q_3$ which is not only computationally more efficient, but requires much less memory storage. 

\subsubsection{Running $Q_1$ and $Q_3$ Estimation Technique}

We quantize error samples based on a non-uniform quantization \cite{nuquan} technique, therefore we need to deal with a fixed and small number of quantization levels (or bins) instead of all data samples to obtain quartiles. Note that in our problem we expect most error samples to accumulate around the origin over time, therefore we use the following compressor function:
\begin{align}
 C(e)=\begin{cases} 
      \frac{1}{1+\exp{\left( -\alpha _1e\right) }}, & e<0 \\  
      \frac{1}{1+\exp{\left( -\alpha _2e\right) }}, & e\geq 0. 
   \end{cases} \nonumber 
\end{align}
This compressor quantizes smaller error samples around $e=0$ with more precision while quantizes error samples farther away less precisely. For this compressor function we have $0<C(e) <1$ and the function $\frac{1}{1+\exp{\left( -\alpha _ie\right) }}$ is called the logistic function in which the parameters $\alpha _1$ and $\alpha_2$ determine the precision of quantization for $e<0$ and $e\geq 0$, respectively. This function is shown in Figure \ref{fig3}.
\begin{figure}
\centering
\includegraphics[scale=.4]{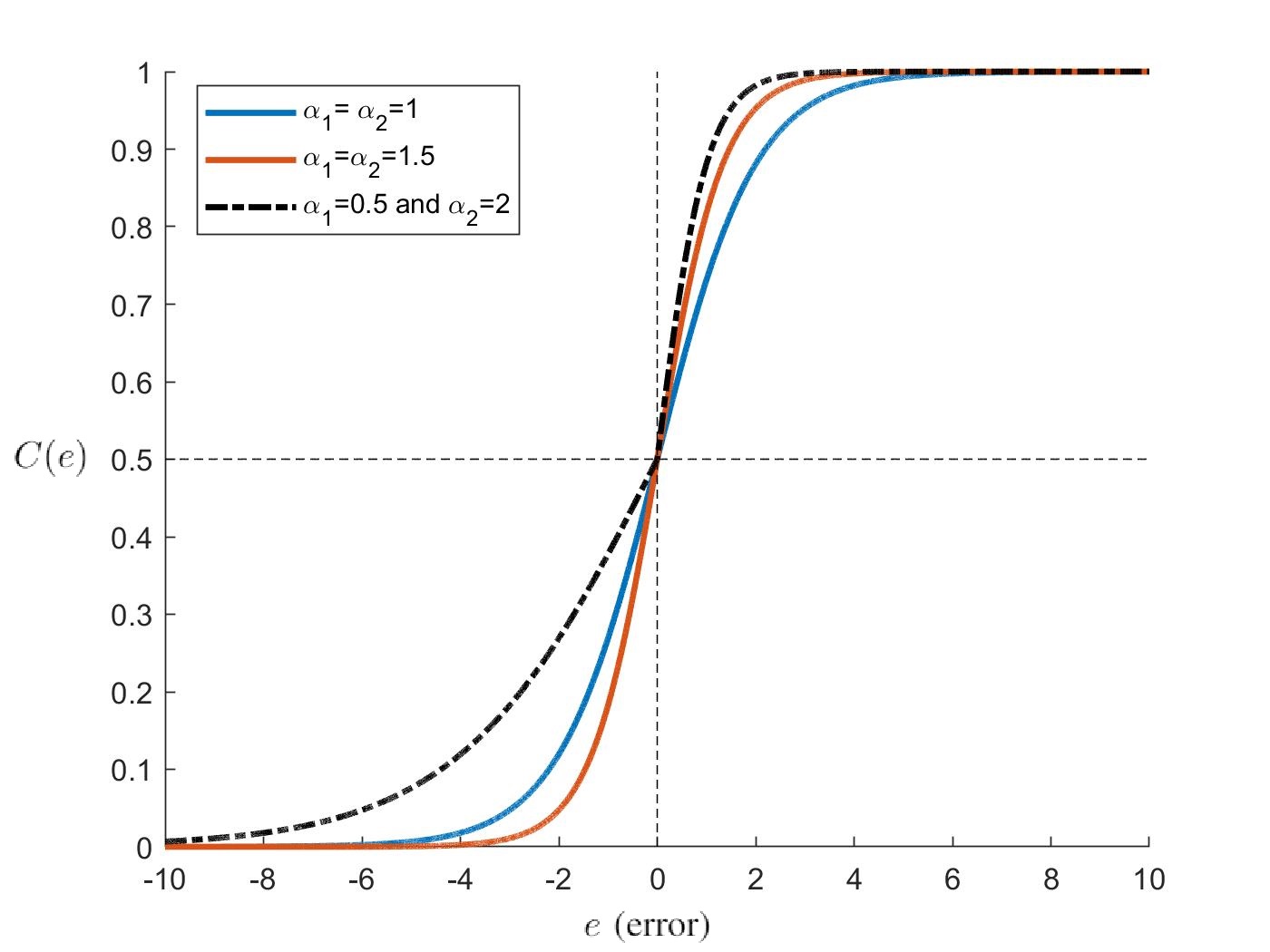}
\caption{Compressor function utilized for non-uniform quantization.}
\label{fig3}
\end{figure}
This compressor function $C(e)$ can be thought of as a cumulative distribution function (CDF) of error PDF (note that it is not exactly the CDF), therefore $0.25$ and $0.75$ on the $y$-axis in Figure \ref{fig3} denote $Q_1$ and $Q_3$, respectively. Note that since we are working with a data stream where we expect error samples to get very close to the origin over time, we can assume the median of the error samples is $0$. 
Each new error sample $e_n$ is compressed by $C(e_n)$, and is put in one of the quantization bins obtained by uniformly dividing $y$-axis in Figure \ref{fig3}. We assign a counter to each of these bins which denotes the number of error samples contained in that bin and bins below. We update these counters and obtain new $Q_1$ and $Q_3$ once a new error sample is available. The advantage of using this method is the fact that we only need to use bin counters to find $Q_1$ and $Q_3$ instead of storing all error samples and sorting them in each time instant. More detailed description of our technique is as follows:\raggedbottom
\begin{itemize}
\item For the first $M$ error samples, find $Q_1$ and $Q_3$ by simply storing and sorting them.
\item Then, adjust $\alpha _1$ and $\alpha _2$ in aforementioned compressor function accordingly, \begin{align}\label{al1}
0.25 = \frac{1}{1+\exp{\left( -\alpha _1Q_1\right) }}~~~ \Longrightarrow ~~~\alpha _1=\frac{-\ln (3)}{Q_1},  
\end{align}and similarly,
\begin{align}\label{al2}
0.75 = \frac{1}{1+\exp{\left( -\alpha _2Q_3\right) }}~~~\Longrightarrow ~~~\alpha _2=\frac{\ln (3)}{Q_3}, 
\end{align}where $\ln$ stands for natural logarithm.
\item Inasmuch as in online regression, first, we expect error samples to get closer to $e=0$ over time which emphasizes that quantization precision is important around the origin, and second, we already have an understanding about the range of error samples given the first $M$ error samples, we can set the maximum acceptable quantization error around the origin (denoted by $\epsilon$). Note that each error sample $e$ is quantized as follows:\begin{align} 
e_q=\Delta .\floor*{\frac{C(e)}{\Delta}},\nonumber 
\end{align}where $\Delta $ and $\floor{x}$ denote the quantization step size and the largest integer less than or equal to $x$, respectively. Now we obtain two quantization step sizes below based on the selected $\epsilon $:\begin{enumerate}
\item If $C=0.5$ is one of the quantization levels then,\begin{align}
&\Delta _1=C(\epsilon )-0.5
\nonumber \\
             ~~~\Longrightarrow ~~~&\Delta _1=\frac{1}{1+\exp{\left( -\alpha _2\epsilon \right) }}-0.5. \nonumber 
             \end{align}
\item If $C=0.5$ is not a quantization level then,\begin{align}
&\Delta _2=0.5-C(-\epsilon )=0.5-\frac{1}{1+\exp{\left( -\alpha _1(-\epsilon) \right) }} \nonumber \\
&~~~\Longrightarrow ~~~\Delta _2=0.5-\frac{1}{1+\exp{\left( \alpha _1\epsilon \right) }}.\nonumber 
\end{align} 
\end{enumerate}An illustration of above cases 1 and 2 is shown in Figure \ref{fig41}. Finally, step size of the quantizer is selected as, \begin{align} 
\Delta =\min\{\frac{1}{M},\Delta _1, \Delta _2\}, \nonumber 
\end{align}and number of quantization levels (or bins) is equal to: \begin{align}
QL=\ceil*{\frac{1}{\Delta}}, \nonumber 
\end{align}where $\ceil{x}$ is the least integer greater than or equal to $x$. 
\item Number of error samples are set to $QL$ and, as depicted in Figure \ref{fig4}, a number from $1$ to $QL$ is given to the counters of quantization bins. \begin{figure}
           \centering
\includegraphics[scale=.85]{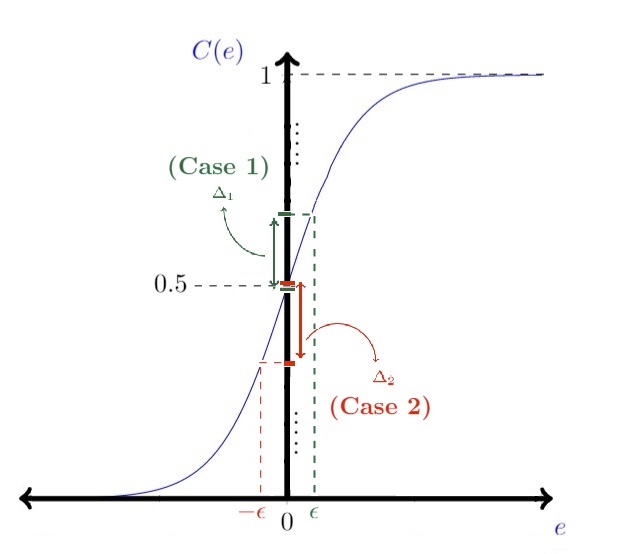}
\caption{$\Delta _1$ and $\Delta _2$ obtained in case 1 ($C=0.5$ is one of the quantization levels) and case 2 ($C=0.5$ is not a quantization level), respectively, based on the maximum acceptable quantization error $\epsilon $ around $e=0$.}
\label{fig41}
\end{figure} \begin{figure}
           \centering
\includegraphics[scale=.85]{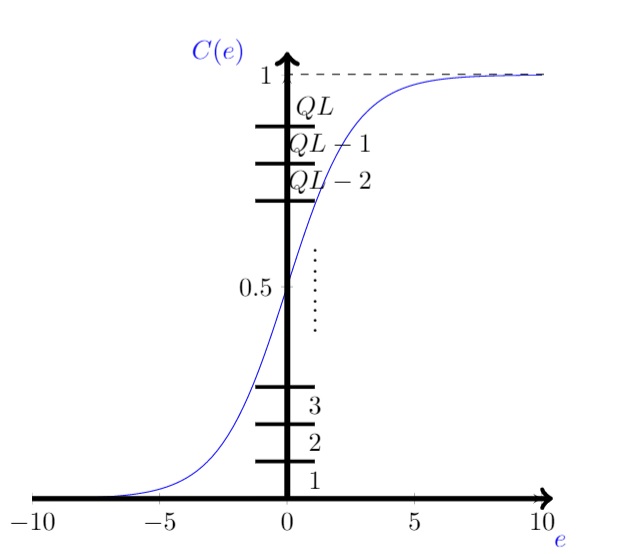}
\caption{A number from $1$ to $QL$ is assigned to counters of quantization bins.}
\label{fig4}
\end{figure}
\item Once $e_{new}$ (a new error sample) is available, we quantize $e_{new}$ and increment the counters related to the bin number $\floor*{\frac{C(e_{new})}{\Delta }}+1$ and all bins above it. Now for each bin the percentage of error samples contained in this bin and bins below it are computed as follows:\begin{align}
\mathrm{Bin}_i=\frac{\mathrm{counter}_i}{\mathrm{number~of~samples}}, \nonumber 
\end{align} where $i\in \{ 1,2,\cdots ,QL-1, QL\}$ is bin index. We can update $Q_1$ and $Q_3$ after receiving every new error sample as we have now the number above for each bin updated. Recall that inasmuch as we expect error samples to get closer to $e=0$ over time, bins related to $Q_1$ and $Q_3$ also get closer to the bin related to $e=0$ (the bin that contains $C(0)=0.5$). In order to obtain updated $Q_1$, we calculate $\ceil*{\frac{\mathrm{Bin}_i}{0.25}}$ and look for the maximum index among all bins for which $\ceil*{\frac{\mathrm{Bin}_i}{0.25}}=1$. Assume this index is found as $I_1$, then $Q_{1,new}$ is calculated as follows:\begin{align} \label{Qnew} 
&\Delta .\left( I_1-1\right) =C(Q_{1,new})=\frac{1}{1+\exp{\left( -\alpha _1.Q_{1,new} \right) }} \nonumber \\
&~~~\Longrightarrow ~~~Q_{1,new}=-\frac{\ln \left( \frac{1}{\Delta .(I_1-1)}-1\right) }{\alpha _1} . 
\end{align} Similarly, $Q_3$ is obtained as follows. Calculate $\floor*{\frac{\mathrm{Bin}_i}{0.75}}$ and look for the minimum index among all bins for which $\floor*{\frac{\mathrm{Bin}_i}{0.75}}=1$. Let this index be $I_3$, then we can obtain $Q_{3,new}$ after similar calculations as follows:\begin{align}
Q_{3,new}=-\frac{\ln \left( \frac{1}{\Delta .(I_3-1)}-1\right) }{\alpha _2} .\nonumber 
\end{align}
\end{itemize}
\vspace{.5cm}\textbf{Remark 2:} Recall that error samples become smaller over time and are mostly around origin, therefore quantization precision would be more important for error samples close to $e=0$. Moreover, we already know that $\alpha _1$ and $\alpha _2$ are significantly important in quantization precision. Consequently, in order to improve our technique when error samples approach to $e=0$ we can update $\alpha _1$ and $\alpha _2$. To this end, we update parameters $\alpha _1$ and $\alpha _2$ whenever $C(Q_1)$ and $C(Q_3)$ are at a specific distance of $C(0)=0.5$. Indeed, when bin index related to $Q_1$ is equal to, \begin{align}
I_{1,c}=\floor*{\frac{0.5}{\Delta }}-\floor*{\beta .QL}, \nonumber 
\end{align} in which we specifically assume $\beta \in (0,0.2)$ to make sure that $I_{1,c}$ is a valid bin index. Afterwards, we should update $\alpha _1$ as follows: \begin{align}
&\Delta .\left( I_{1,c}-1\right) =C(Q_{1,new}) \nonumber \\
&~~~\overset{(\ref{Qnew})}{\Longrightarrow}~~~Q_{1,new}=-\frac{\ln \left( \frac{1}{\Delta .(I_{1,c}-1)}-1\right) }{\alpha _1}  \nonumber \\
&~~~\overset{(\ref{al1})}{\Longrightarrow}~~~\alpha _{1,new} = -\frac{\ln (3)}{Q_{1,new}}. \nonumber 
\end{align}
Similarly, when bin index of $Q_3$ is equal to the following, \begin{align}
I_{3,c}=\floor*{\frac{0.5}{\Delta }}+\floor*{\beta .QL}, \nonumber 
\end{align} new $\alpha _2$ is obtained as follows: 
\begin{align}
Q_{3,new}=-\frac{\ln \left( \frac{1}{\Delta .(I_{3,c}-1)}-1\right) }{\alpha _2} ~~~\overset{(\ref{al2})}{\Longrightarrow} ~~~\alpha _{2,new} =\frac{\ln (3)}{Q_{3,new}}. \nonumber 
\end{align}
Note that once bin index related to $Q_1$ ($Q_3$) is equal to $I_{1,c}$ ($I_{3,c}$), we update both $\alpha _1$ and $\alpha _2$ and reset number of error samples to $QL$ and assign the counters of quantization bins a number from 1 to $QL$ as depicted in Figure \ref{fig4}. Recall that once we obtain $Q_1$ and $Q_3$ in each time instant we can obtain lower and upper extremes as depicted in Figure \ref{fig2} and use them to determine whether the current error sample is abnormally large (a major outlier) or not. The algorithm of the proposed error samples running quartile estimation technique is shown in Algorithm 1.

The aforementioned algorithm for running quartile estimation is the key tool that we will use in the next parts of this section.
\subsubsection{Trimming MEE}
A major outlier in error samples denotes that either an abnormally large label (significantly different from most of the labels) or an abnormally large noise in label has occurred. In either way, it seems that the benefit of just ignoring this major outlier error sample and not incorporating it in the learning process would be more than using it. This is obvious inasmuch as:\begin{itemize}
\item If the label itself is significantly different from rest of the labels, learning based on that will update the parameters in favor of this major outlier not most of the data, therefore the resultant parameters will not be promising.
\item If the label has been corrupted by an abnormally large noise, updating parameters based on this heavily corrupted data sample will be very misleading.
\end{itemize} 
Throughout this paper, we stop incorporating major outliers into online regression by detecting and eliminating them from MEE using our proposed running quartile estimation technique and finding lower and upper extremes as depicted in Figure \ref{fig2}. This can be interpreted as ignoring heavy part of the tail in error PDF and only use lighter part for learning.
\subsubsection{Locating Error PDF at the Origin in MEE}
As discussed earlier within introduction, simply adding sample mean of labels as a bias to the learned system output or using fiducial points in order to locate error PDF at the origin may be problematic. In this subsection, we deploy our running quartile estimation technique to address this problem by modifying first method in which we add a \textit{proper} bias to the learned system output. First, we discuss the possible problem of adding a bias in more details here. Remember that outliers can be very misleading in sample mean approximation inasmuch as very large values can significantly shift sample mean approximation away from actual mean in each time instant. Outliers in error samples usually arise from heavy-tailed distribution (of noise in label or label itself). 
A heavy-tailed random variable is defined as follows:\\
\begin{algorithm}[H]
 \caption{Our proposed error samples running quartile estimation technique}
 \begin{algorithmic}[1]
 \renewcommand{\algorithmicrequire}{\textbf{Input:}}
 \renewcommand{\algorithmicensure}{\textbf{Outputs:}}
 \REQUIRE $\{e_n\}$
 \ENSURE  $Q_{1,n} \mathrm{~and~} Q_{3,n}$ (first and third quartiles at time instant $n$)
 \\ \textit{Initialisation} : M and $\beta $
   \IF {($n< \mathrm{M}$)}
   
  \STATE $\mathrm{Sort~} \{ e_1,e_2,\cdots ,e_n\} \mathrm{~and~ obtain~ Q_{1,n} \mathrm{~and~} Q_{3,n}}$
  
  \ELSIF{($n= \mathrm{M}$)}
  
  \STATE $\mathrm{Sort~} \{ e_1,e_2,\cdots ,e_M\} \mathrm{~and~ obtain~} Q_{1,M} \mathrm{~and~} Q_{3,M}$
  
  \STATE Calculate $\alpha _1$ and $\alpha _2$: $$\alpha _1=\frac{-\ln (3)}{Q_{1,\mathrm{M}}} \mathrm{~~~and~~~} \alpha _2=\frac{\ln (3)}{Q_{3,\mathrm{M}}}$$
  
  \STATE Select $\epsilon$
  
  \STATE Calculate $\Delta _1$ and $\Delta _2$: \begin{align} &\Delta _1=\frac{1}{1+\exp{\left( -\alpha _2\epsilon \right) }}-0.5 \mathrm{~~~and~~~} \nonumber \\
 &\Delta _2=0.5-\frac{1}{1+\exp{\left( \alpha _1\epsilon \right) }} \nonumber \end{align}
  
  \STATE Select step size of the quantizer: $$ \Delta =\min\{\frac{1}{M},\Delta _1, \Delta _2\}$$
  
  \STATE  Set NES (number of error samples) = QL (number of quntization levels or bins) where QL=$\ceil*{\frac{1}{\Delta}}$
  
  \STATE Create a vector $\mathrm{COUNTER}_{1\times QL}$ containing counters to be updated with initialisation $\mathrm{COUNTER}(i) = i,~i=1,2,\cdots ,QL$    
  
  \STATE Calculate $ I_{1,c}$ and $ I_{3,c}$: $$ I_{1,c}=\floor*{\frac{0.5}{\Delta }}-\floor*{\beta .QL} \mathrm{~~~and~~~} I_{3,c}=\floor*{\frac{0.5}{\Delta }}+\floor*{\beta .QL}$$
  
  \ELSE 

\STATE $\mathrm{COUNTER}(j)=\mathrm{COUNTER}(j)+1$ for $j\geq \floor*{\frac{C(e_{new})}{\Delta }}+1$

  \STATE Compute the percentage of error samples contained in each bin $i$ and bins below it: $$\mathrm{Bin}_i=\frac{\mathrm{COUNTER}(i)}{\mathrm{NES}} $$
  
  \STATE $\mathrm{NES} = \mathrm{NES} +1$
  
  \STATE Find $I_1=\mathrm{\underset{i}{argmax}}\left(\ceil*{\frac{\mathrm{Bin}_i}{0.25}}==1 \right) \mathrm{~and~} I_3=\mathrm{\underset{i}{argmin}}\left( \floor*{\frac{\mathrm{Bin}_i}{0.75}}==1\right) $
  
   \STATE  Calculate $Q_{1,n}$ and $Q_{3,n}$: \begin{align}
   &Q_{1,n}=-\frac{\ln \left( \frac{1}{\Delta .(I_1-1)}-1\right) }{\alpha _1} \mathrm{~~~and~~~} \nonumber \\
   &Q_{3,n}=-\frac{\ln \left( \frac{1}{\Delta .(I_3-1)}-1\right) }{\alpha _2} \nonumber \end{align}

   \IF{($I_1==I_{1,c}$ or $I_3==I_{3,c}$)}
  
  \STATE Go back to line 5 and update $\alpha _1$ and $\alpha _2$ (replace $Q_{1,M}$ and $Q_{3,M}$ with $Q_{1,n}$ and $Q_{3,n}$, respectively)
  \STATE Go to lines 9 and 10
  
  \ELSE
  
  \STATE Continue
  
  \ENDIF
  
  \ENDIF
 \RETURN $Q_{1,n}$ and $Q_{3,n}$ 
 \end{algorithmic} 
 \end{algorithm}  
\textbf{Definition 1 \cite{20.1}:} Random variable $E$ is said to be heavy-tailed if its CDF $F(e)=P(E\leq e)$ has the following property for all $\mu >0$:
\begin{align}
\lim _{e\rightarrow \infty}sup\frac{1-F(e)}{\exp{(-\mu e)}}=\infty .\nonumber 
\end{align}This means that a distribution is heavy-tailed if its tail is heavier than that of any exponential distribution. For these kind of distributions, law of large numbers (LLN) and central limit theorem (CLT) do not hold and sample mean approximation may not converge to the actual mean \cite{20.1}. Some moments of heavy-tailed distribution may not exist (even mean and variance), for instance, Cauchy distribution is a heavy-tailed distribution for which neither mean nor variance exist and obviously sample mean approximation does not make sense. Interestingly, even if all moments exist for a heavy-tailed distribution and sample mean converges to actual mean, the convergence will be very slower than that of a light-tailed distribution. This can be easily shown, however first we need the following Lemma to state and prove Corollary 2.

\begin{lemma} \cite{20.1}: Following two statements are equivalent for random variable $E$:
\begin{itemize}
\item $E$ is heavy-tailed 
\item $M_E(s)\triangleq E\{ e^{sE}\}=\infty $ for all $s>0$
\end{itemize}
where $M_E(s)$ denotes moment generating function of $E$.
\end{lemma}
\begin{proof}
Provided in \cite{20.1}.
\end{proof}
Next, we use above lemma to prove following corollary.
\begin{corollary}
Consider $n$ independent heavy-tailed random variables $E_1,E_2,\cdots ,E_n$, then sum random variable $S_n=E_1+E_2+\cdots +E_n$ is also a heavy-tailed random variable.
\end{corollary}
\begin{proof}
Writing the moment generating function of the sum random variable, we have:
\begin{align}
M_{S_n}(s) \overset{(a)}{=} M_{E_1}\times M_{E_2}\times \cdots \times M_{E_n} \overset{(b)}{=}\infty \mathrm{~~~for~all~s>0},\nonumber 
\end{align}
where $(a)$ is due to the independence of random variables $E_1,E_2,\cdots ,E_n$ from each other and  $(b)$ is because of the fact that these random variables are heavy-tailed and consequently based on Lemma 1 their moment generating function is $\infty $. Therefore, sum random variable $S_n$ is also heavy-tailed based on Lemma 1.  
\end{proof}

Now, we show even if all moments exist for a heavy-tailed distribution and sample mean converges to actual mean, the convergence will be very slow. Consider $n$ independent and identically distributed ($i.i.d.$) heavy-tailed random variables $E_i,~i=1,2,\cdots ,n$, with mean $m $ and variance $\sigma ^2$. Similar to the proof of the corollary 2 we can show that $Z_n=\frac{S_n-nm}{\sqrt{n}}$ also has a heavy-tailed distribution with mean $0$ and variance $\sigma ^2$. Therefore, any realization $z_n$ of random variable $Z_n$ at each time instant $n$ can be written as follows:
\begin{align}
\label{HT_CLT}
z_n=\frac{s_n-nm}{\sqrt{n}}~~~\Longrightarrow ~~~\frac{s_n}{n}=m+\frac{z_n}{\sqrt{n}}, 
\end{align}
where $s_n$ is a realization of heavy-tailed distribution $S_n$ with mean $nm$ and variance $n\sigma ^2$. $\frac{s_n}{n}$ is sample mean approximation that can be heavily distorted around the true mean $m$ inasmuch as samples from a heavy-tailed distribution $Z_n$ contain very large values. Compare this with what we have for a light-tailed distribution. Consider $n$ $i.i.d.$ light-tailed random variables $E_i,~i=1,2,\cdots ,n$, with mean $m$ and variance $\sigma ^2$. Based on CLT for large $n$ we have,
\begin{align}
\label{CLT}
\frac{s_n}{n}\approx  m+\frac{z}{\sqrt{n}}, 
\end{align}
where $s_n$ and $z$ are realizations of sum random variable $S_n$ with mean $nm$ and variance $n\sigma ^2$ and a Gaussian distribution with mean $0$ and variance $\sigma ^2$, respectively. Comparing (\ref{HT_CLT}) and (\ref{CLT}) we can see although sample mean $\frac{s_n}{n}$ converges to actual mean $m$ as $n$ goes to $\infty $ for both of these heavy-tailed and light-tailed scenarios, convergence in light-tailed scenario is much faster. This is because of the fact that samples of a light-tailed distribution never differ significantly from the mean while samples from a heavy-tailed distribution contain very large samples or outliers, therefore $\frac{z}{\sqrt{n}}$ approaches much faster to zero than $\frac{z_n}{\sqrt{n}}$ as $n$ increases.

Lets get back to our main problem in this subsection. It is clear now why adding sample mean of labels to system output as a bias term in order to locate error PDF at the origin may be problematic: this sample mean approximation may not be accurate if label or noise in label has a heavy-tailed distribution. As discussed above, an important characterization of light-tailed distributions is that the samples are concentrated around the mean of the distribution not far away from that and are mostly of a
similar size. In contrast, sampling from a heavy-tailed distribution yields a few very large samples in addition to many small ones that can tremendously dominate the sum (or equivalently sample mean approximation). Inasmuch as we clean MEE and do not incorporate major outliers (or very large error samples) in learning process, we can assume them eliminated from error PDF which means we lighten the possible heavy-tailed error PDF and consequently sample mean approximation is not problematic anymore. Recall that we use again our proposed running quartile estimation technique to detect and eliminate major outliers from sample mean approximation.
\begin{algorithm}[H]
 \caption{Trimmed MEE}
 \begin{algorithmic}[1]
 \renewcommand{\algorithmicrequire}{\textbf{Inputs:}}
 \renewcommand{\algorithmicensure}{\textbf{Output:}}
 \REQUIRE $\{\textbf{x}_n,d_n\}$
 \ENSURE  $\textbf{w}_n$ and $\mathrm{BIAS}_n$
 \\ \textit{Initialisation} : M, $\beta,~\mu$, $\sigma $, N, $\mathrm{counter}_{NO} = 0$ (number of non-outliers), $\bar{d}=0$, $\bar{\textbf{x}}=\textbf{0}$, $\mathrm{BIAS_0} = 0$ and $\textbf{w}_0=\textbf{0}$ 
  \FOR {$\mathrm{each~iteration}~n$}
  \STATE $e_n=d_n-(\textbf{x}_n^T\textbf{w}_{n-1}+\mathrm{BIAS}_{n-1})$
  \STATE Use algorithm 1 to obtain $Q_{1,n}$ and $Q_{3,n}$
  \STATE $IQR_n=Q_{3,n}-Q_{1,n}$
  \STATE $UE_n$ =$Q_{3,n}+3\times IQR_n$ (upper Extreme at time instant $n$)
  \STATE $LE_n$ =$Q_{1,n}-3\times IQR_n$ (lower Extreme at time instant $n$)
  \IF {($LE_n\leq e_n\leq UE_n$)}
  \STATE Calculate,
  \begin{align}
\nonumber \\
 &\bar{d}=\frac{d_n+(\mathrm{counter}_{NO})\bar{d}}{\mathrm{counter}_{NO}+1},  ~~~\bar{\textbf{x}}=\frac{\textbf{x}_n+(\mathrm{counter}_{NO})\bar{\textbf{x}}}{\mathrm{counter}_{NO}+1}, \nonumber \\
 \nonumber \\
 &\mathrm{counter}_{NO} = \mathrm{counter}_{NO} +1, \nonumber  \\
 \nonumber 
 \end{align}
 \begin{dmath*}
 ~~~~~\boldsymbol{\nabla }\hat{I}_2^{(n)}(E) = \frac{1}{N^2\sigma ^2}{\sum \sum }_{\underset{LE_{n}\leq e_{n-i},e_{n-j} \leq UE_{n}}{0\leq i,j\leq N-1,}}\Bigg[ G_{\sigma }(e_{n-i}-e_{n-j}) (e_{n-i}-e_{n-j})(\textbf{x}_{n-i}-\textbf{x}_{n-j}) \Bigg] , 
\end{dmath*}
\STATE $\textbf{w}_n=\textbf{w}_{n-1}+\mu \boldsymbol{\nabla }\hat{I}_2^{(n)}(E) ,~~~ \mathrm{BIAS}_n=\bar{d}-\bar{\textbf{x}}^T\textbf{w}_n$
  \ELSE
  \STATE $\textbf{w}_n=\textbf{w}_{n-1},~~~\mathrm{BIAS}_n=\mathrm{BIAS}_{n-1}$
  \ENDIF
  \ENDFOR
 \RETURN $\textbf{w}_n$ and $\mathrm{BIAS}_n$ 
 \end{algorithmic} 
 \end{algorithm}
Eventually, to sum up this section, our proposed algorithm for online linear regression, called \textit{Trimmed} MEE, is shown in Algorithm 2. At each time instant $n$ the output of the system is learned as follows:
\begin{align}
y_n=\textbf{x}_n^T\textbf{w}_{n-1}+\mathrm{BIAS}_{n-1},  \nonumber
\end{align}
where $\mathrm{BIAS}_{n-1}$ is the sample mean calculated at time instant $n-1$ from samples of the tail-lightened but non-centered error PDF. The tail-lightened and centered error is denoted by (note that major outliers are not considered in following calculations):
\begin{align}
\label{efinal}
e_n=d_n-y_n=d_n-(\textbf{x}_n^T\textbf{w}_{n-1}+\mathrm{BIAS}_{n-1}).      
\end{align}

 In the next section, simulation results are shown and discussed.
\section{Simulation Results}
Throughout this section we consider the linear adaptive filtering problem illustrated in Figure \ref{fig6} which can be viewed as an online linear regression (for MEE and Trimmed MEE the BIAS block is also considered). In this Figure we have $\tilde{d}_n=\textbf{x}_n^T\textbf{w}_{opt}$ and $\tilde{y}_n=\textbf{x}_n^T\textbf{w}_{n-1}$ where $\textbf{x}_n=\big[ x_n,x_{n-1},\cdots ,x_{n-L+1}\big] ^T$ and $\textbf{w}_{n-1}=\big[ w_{n-1}^{(1)},w_{n-1}^{(2)},\cdots ,w_{n-1}^{(L)}\big] ^T$ denote the input vector to the system at time instant $n$ and system parameters estimated at time instant $n-1$, respectively. We employ algorithm 2 to adapt this linear system and obtain its parameter vector $\textbf{w}$ and $\mathrm{BIAS}$ in each time instant. Moreover, for the sake of computational simplicity, we assume a fixed kernel bandwidth $\sigma $ and also drop the expectation operator from both (\ref{correntropy}) and (\ref{L1}) to use the following online cost functions for stochastic MCC and stochastic MEE, respectively:
\begin{align}
&v_s(E)=G_{\sigma }(E) ~~~\Longrightarrow ~~~\hat{v}_s^{(n)}(E)=G_{\sigma }(e_n) ,\nonumber \\
&I_{2,s}(E)=p(E) ~~~\Longrightarrow ~~~\nonumber \\
&~~~~~~~~~~~~~\hat{I}_{2,s}^{(n)}(E)=p(e_n) \overset{\mathrm{Parzen}}{=}\frac{1}{N}\sum _{i=0}^{N-1}G_{\sigma }\left( e_n-e_{n-i} \right) ,\nonumber 
\end{align}
 \begin{figure*}
\centering
\includegraphics[scale=.7]{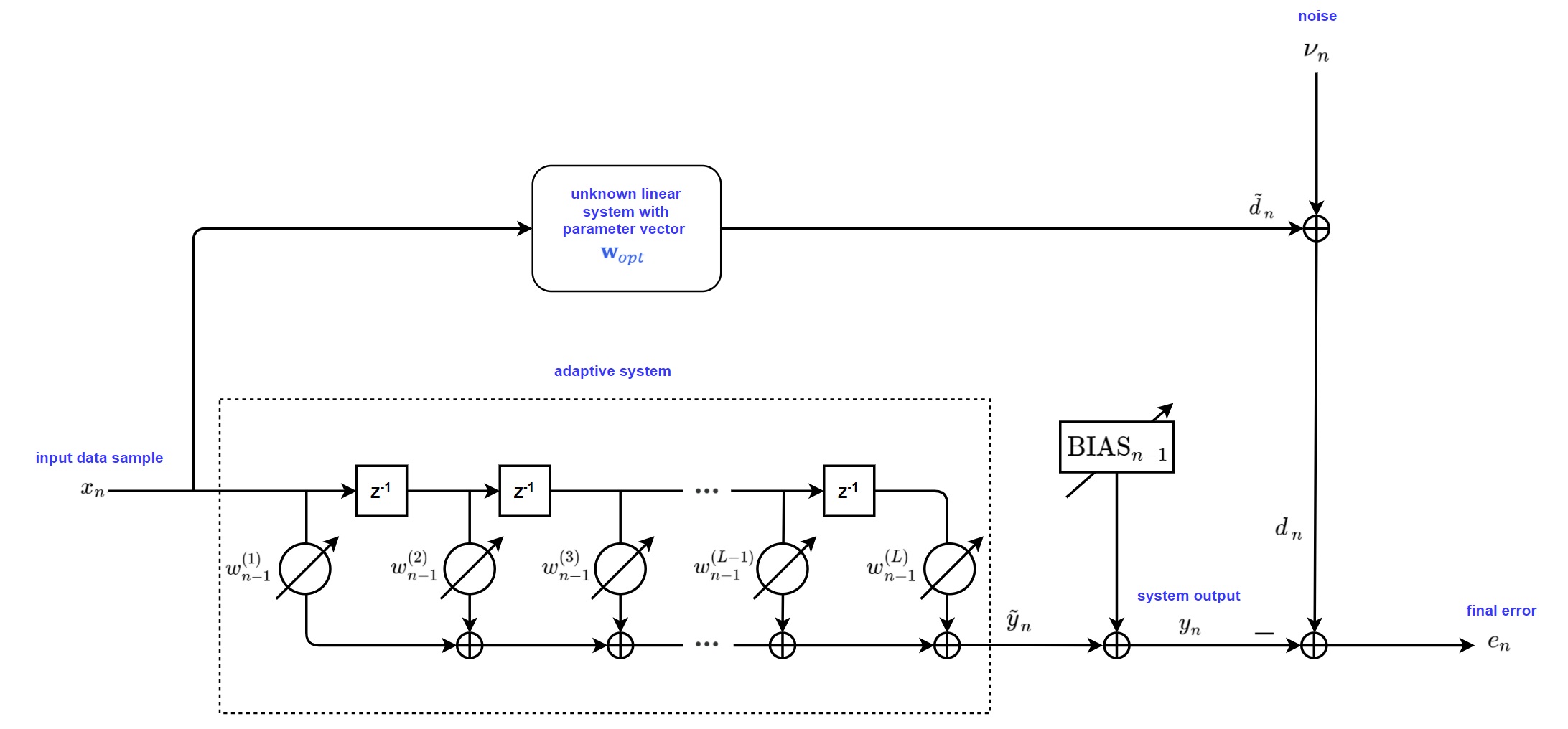}
\caption{Linear adaptive filtering.}
\label{fig6}
\end{figure*}
and consequently we have the following online cost function for stochastic MEEF with $M$ fiducial point:
\begin{align}
\hat{I}_{2,F,s}^{(n)}(E)&=\frac{1}{N+M}\sum _{i=0}^{N+M-1}G_{\sigma }\left( e_n-e_{n-i} \right) \nonumber \\
&=\frac{N}{N+M}\hat{I}_{2,s}^{(n)}(E) +\frac{M}{N+M}\hat{v}_s^{(n)}(E). \nonumber 
\end{align}
In addition, in our simulations throughout this section we draw input data samples $x_i$ in Figure \ref{fig6} from a standard normal distribution, i.e., $x_i\sim \mathcal{N} \left( 0,1\right)$. The unknown system parameter vector (or optimum weight vector to be learned) is generated randomly as a unit vector $\textbf{w}_{opt}\in \mathcal{R}^L$ where $\parallel \textbf{w}_{opt}\parallel _2=1$ and $L=5$. We also measure misalignment at time instant $n$ for an algorithm based on the normalized mean-square deviation (NMSD) as follows:
\begin{align} 
\mathrm{misalignment}_n&=20\log _{10}\left( \frac{\parallel \textbf{w}_n-\textbf{w}_{opt}\parallel }{\parallel \textbf{w}_{opt}\parallel }\right) \nonumber \\
&=20\log _{10}\left( \parallel \textbf{w}_n-\textbf{w}_{opt}\parallel \right) . 
\nonumber 
\end{align}
\vspace{-.7cm}\subsection{Comparison of MCC and MEE}
Recall that there is a difference between MCC and MEE cost functions as derived in (\ref{MCCvsMEE}) which means more difference between these cost functions will cause more difference in the results obtained based on them. We also expect MEE, as discussed earlier, to be a more comprehensive cost function than MCC and have a superior performance compared with that, therefore using fiducial points in MEEF, which alter MEE cost function to a weighted combination of MEE and MCC cost functions, will badly affect MEE performance whenever the difference between MEE and MCC is substantial. For instance, as shown in \cite{8.5}, for exponential noise the difference between two cost functions can become considerable. We can conclude then combination of MEE with MCC as MEEF is expected not to be helpful. This fact is shown in Figure \ref{fig7} for different values of learning rate $\mu $ and kernel bandwidth $\sigma $. More precisely, note that two important factors to evaluate a learning algorithm are its convergence rate and steady state misalignment. Given a specific pair $(\mu ,\sigma )$ we run each algorithm and obtain steady state misalignment for that as the sample mean of the misalignments within the last 200 time instants. Moreover, given a specific pair $(\mu ,\sigma )$ we denote convergence rate for each algorithm by $i_{conv.}$ and define it as the first time instant (or iteration) for which we have the following:
\begin{align}
\mathrm{misalignment}_{i_{conv.}}\leq \mathrm{~steady ~state~ misalignment~ +~ 2}. \nonumber 
\end{align}
Now, recall that label (or desired signal) at time instant $n$ is modeled as $d_n=\textbf{x}_n^T\textbf{w}_{opt}+\nu_n.$ We assume an exponential noise (which is not heavy-tailed), i.e., $\nu \sim Exp(\lambda )$, with following PDF:
\begin{align}
f(\nu )=\lambda \exp \left( -\lambda \nu \right) , ~\nu \geq 0~~\mathrm{and}~~f(\nu )=0,~\mathrm{otherwise}\nonumber 
\end{align}
and 30dB signal to noise ratio (SNR) calculated as follows:
\begin{align}
\mathrm{SNR}=10\log _{10}\left( \frac{E\left \{ \big[\textbf{x}_n^T\textbf{w}_{opt}\big] ^2\right \} }{E\{ \nu _n^2 \}}\right) , \nonumber 
\end{align}
which results in $\lambda = \sqrt{2000}$. Achievable bounds in Figure \ref{fig7} are obtained from convex hull of all points resulted from pairs $(\mu , \sigma )$ such that $\mu \in [0.05:0.005:0.1]$ and $\sigma \in [0.2:0.1:1.4]$ where for each pair $(\mu ,\sigma )$ the convergence rate and steady state misalignment results are obtained by averaging learning curves over 50 Monte Carlo simulations. As seen in Figure \ref{fig7}, MEE outperforms MCC significantly for exponential noise. This can be translated to the fact that using MEEF for an environment corrupted by this exponential noise is not helpful inasmuch as it degrades MEE performance by incorporating MCC into it. This is shown in Figure \ref{fig8}-a in which learning curves (averaged over 200 Monte Carlo simulations) of MEE, MCC and MEEF are compared. As seen in this Figure, as we increase number of fiducial points, the performance of the MEE algorithm deteriorates.  
\begin{figure}[t]
\centering
\includegraphics[scale=.5]{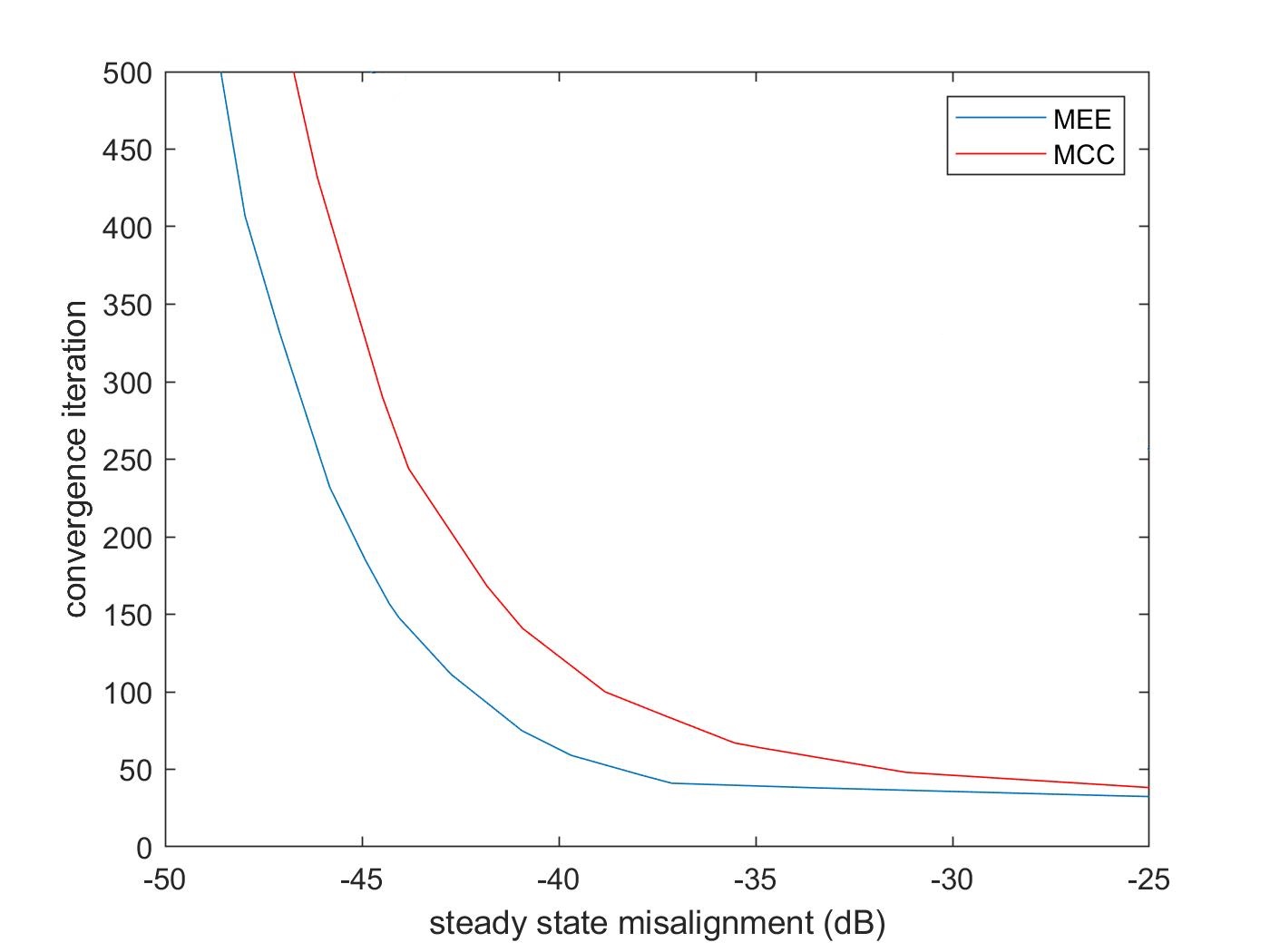}
\caption{Convergence iteration vs. steady state misalignment for MEE and MCC algorithms obtained from different pairs $(\mu , \sigma )$ such that $\mu \in [0.05:0.005:0.1]$ and $\sigma \in [0.2:0.1:1.4]$ (in presence of exponential noise and 30dB SNR). }
\label{fig7}
\end{figure}

\begin{figure}[t]
\centering
\subfloat{\textbf{a.}}{
	\includegraphics[width=0.45\textwidth]{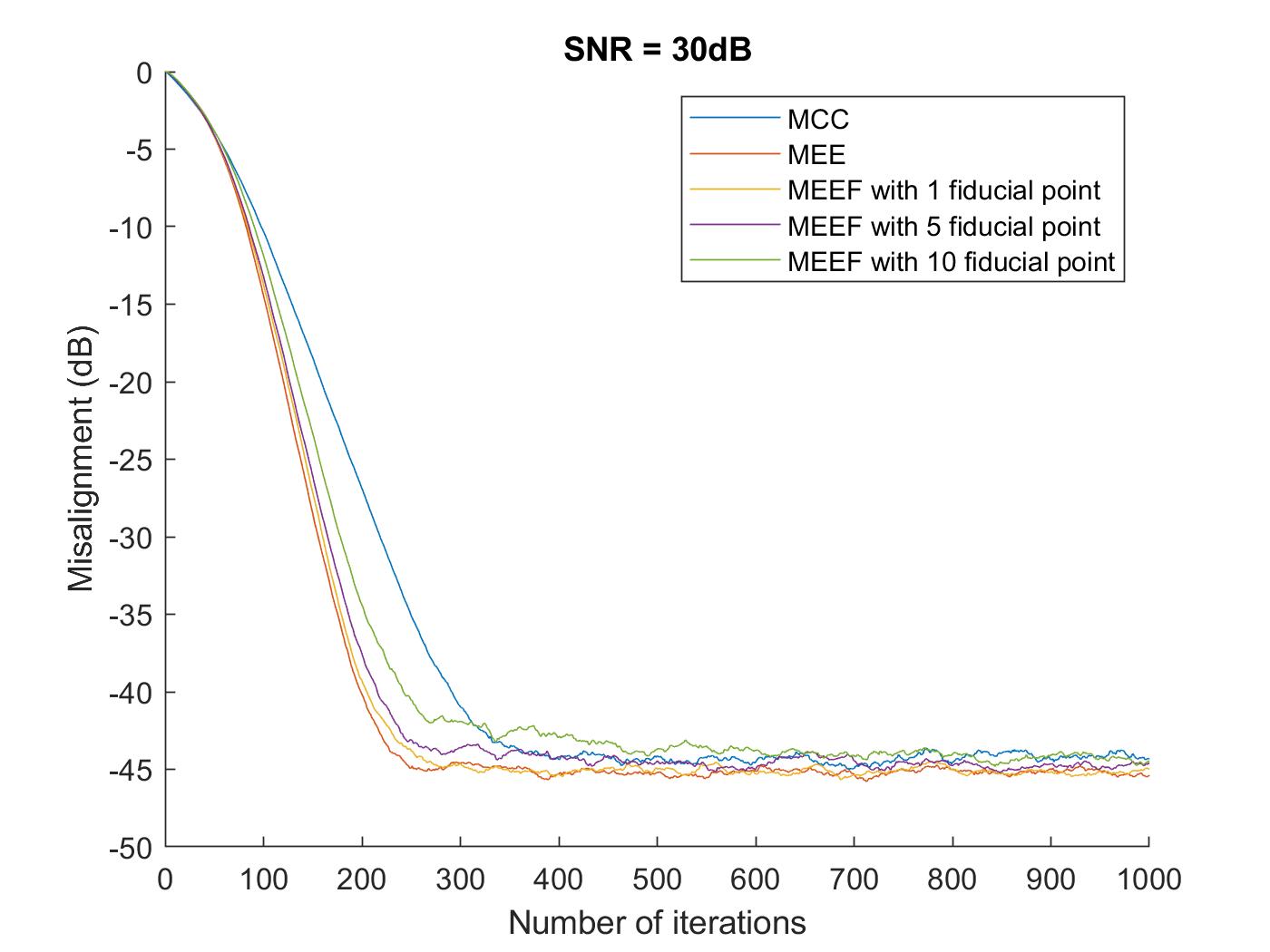} } 
\subfloat{\textbf{b.}}{
	\includegraphics[width=0.45\textwidth]{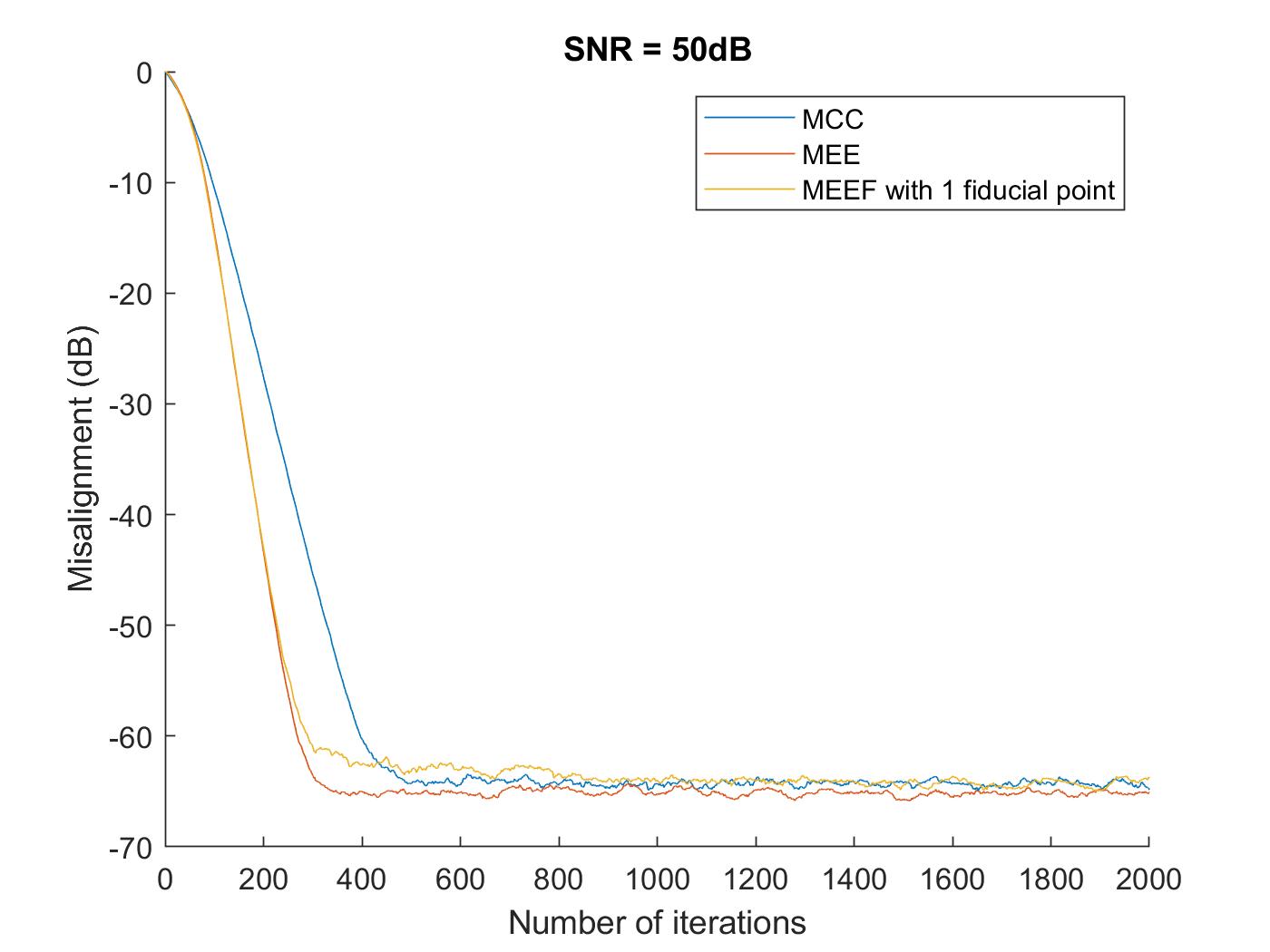} } 
\caption{Learning curves of different algorithms under exponential noise averaged over 200 Monte Carlo simulations with $\mu = 0.05$ and $\sigma =1$ (a. SNR=30dB and b. SNR=50dB).}
\label{fig8}
\end{figure}
It is worth mentioning that, as discussed in \cite{8.5}, when the parameter of the exponential noise $\lambda $ is increased the difference between MCC and MEE cost functions becomes larger as well and therefore we expect more destructive effect of incorporating MCC into MEE as MEEF. This is shown in Figure \ref{fig8}-b where we assume 50dB SNR which results in a larger $\lambda $ ($\lambda =\sqrt{2\times 10^5}$) and as depicted in this Figure even one fiducial point corrupts MEE performance significantly. Convergence iteration versus steady state misalignment achievable bounds for MEE, MEEF with 1 fiducial point and MCC are illustrated in Figure \ref{fig9} for 50dB SNR. As seen in this figure, MEE outperforms both MCC and MEEF.
\begin{figure}[t]
\centering
\includegraphics[scale=.5]{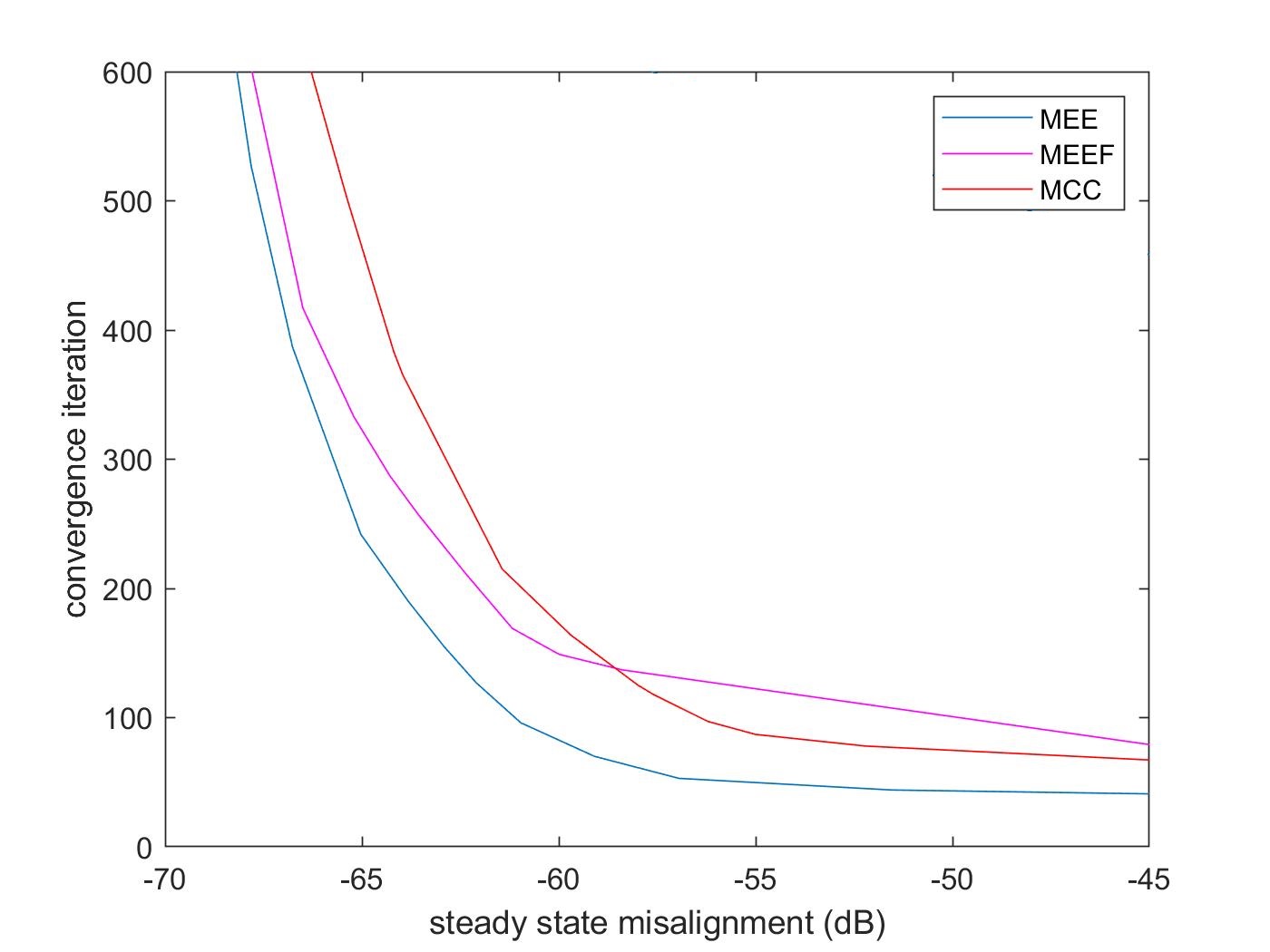}
\caption{Convergence iteration vs. steady state misalignment for MEE, MEEF with 1 fiducial point and MCC algorithms obtained from different pairs $(\mu , \sigma )$ such that $\mu \in [0.05:0.005:0.1]$ and $\sigma \in [0.2:0.1:1.4]$ (in presence of exponential noise and 50dB SNR). }
\label{fig9}
\end{figure}
\begin{figure}[t]
\centering
\includegraphics[scale=.16]{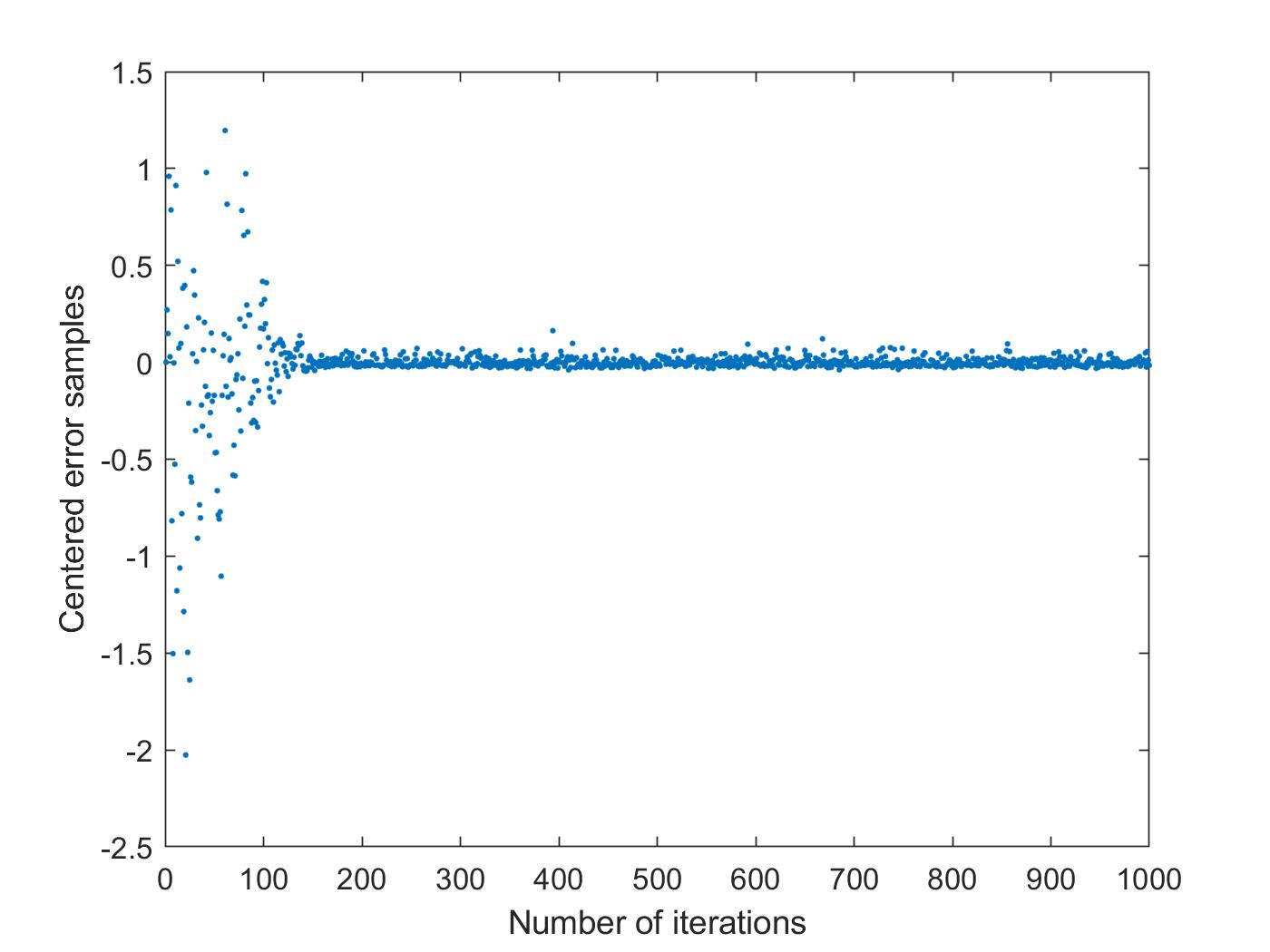}
\caption{Scatter plot of final error samples of MEE algorithm of Figure \ref{fig8}-a.}
\label{fig10}
\end{figure}
\subsection{Performance analysis of our proposed error samples running quartile estimation technique}
The ground truth for evaluation of our proposed error samples running quartile estimation technique are the results obtained for $Q_{1,n}$ and $Q_{3,n}$ at time instant $n$ from order statistics by simply sorting final error samples (centered error samples) in each time instant. We consider final error samples (\ref{efinal}) resulted from running MEE algorithm in Figure \ref{fig8}-a. The scatter plot of these error samples are shown in Figure \ref{fig10}. As seen in this Figure, centered error samples are accumulated around the origin with time (and consequently $Q_1$ and $Q_3$ also get closer to the origin with time) which verifies our assumption in algorithm 1 that median of error samples is assumed zero. It is worth mentioning that as expected we do not see major outliers in centered error samples in Figure \ref{fig10} inasmuch as exponential noise is not heavy-tailed. Now we plot $Q_{1,n}$ and $Q_{3,n}$ of these centered error samples in Figure \ref{fig11}-a and Figure \ref{fig11}-b, respectively obtained from both order statistics and our proposed technique with $M=100$ and $\epsilon = 0.01$ (Algorithm 1). As seen in this Figure, the difference between $Q_{1,n}$ ($Q_{3,n}$) obtained from sorting and that obtained from our technique is negligible.
\begin{figure}[t]
\centering
\subfloat{\textbf{a.}}{
	\includegraphics[width=0.45\textwidth]{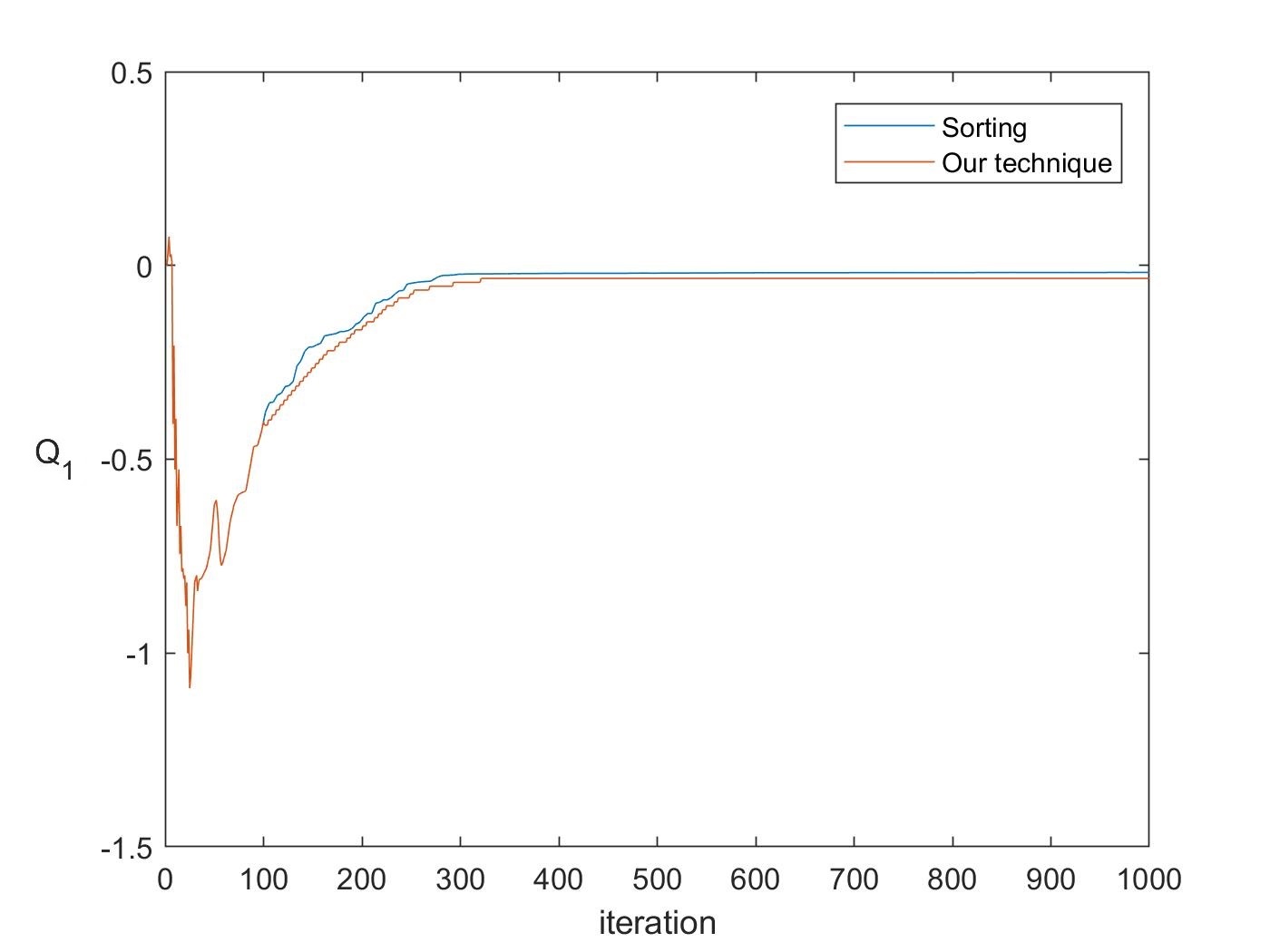} } 
\subfloat{\textbf{b.}}{
	\includegraphics[width=0.45\textwidth]{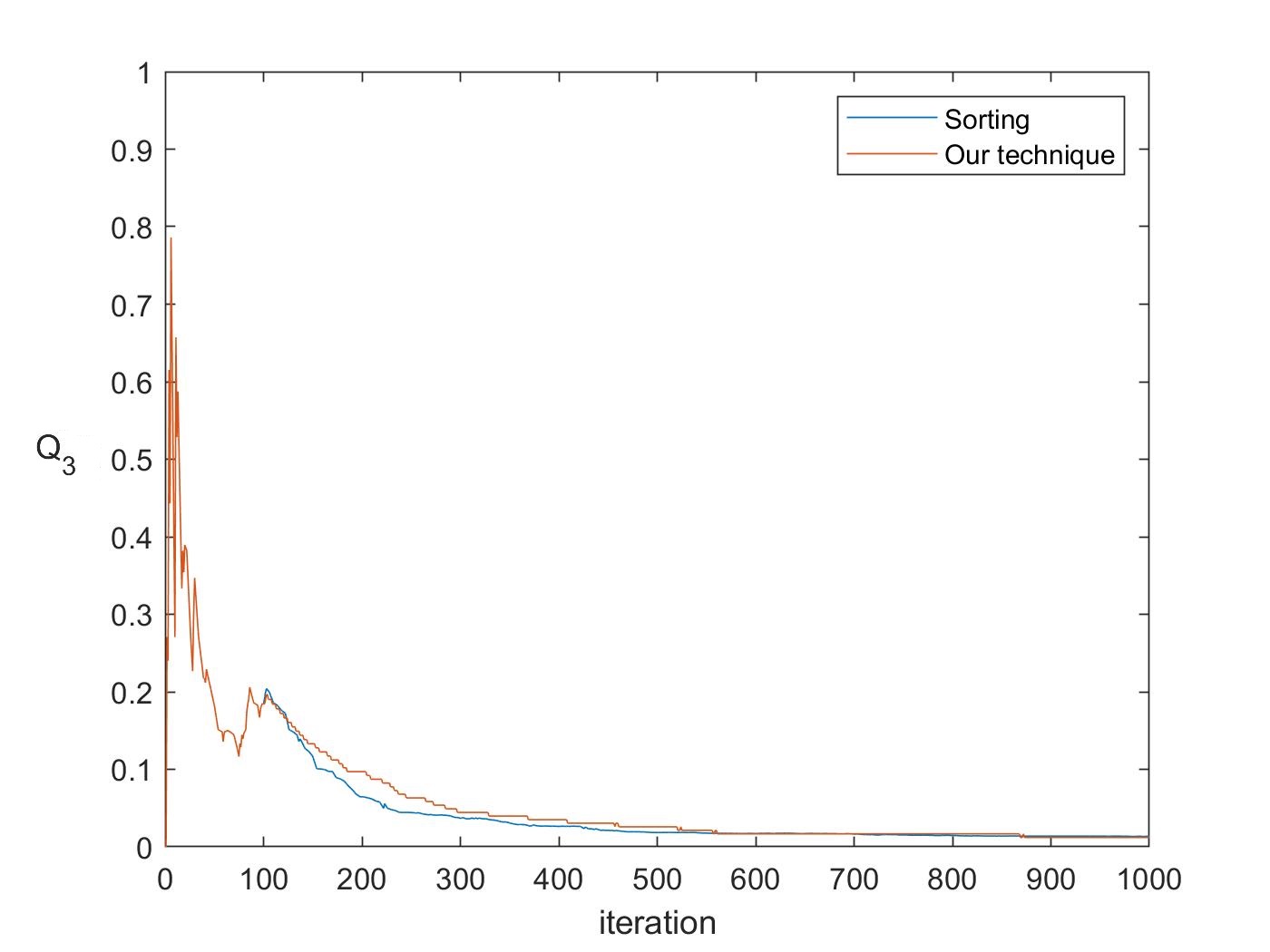} } 
\caption{Lower quartile $Q_1$ and upper quartile $Q_3$ of error samples of Figure 9 estimated at each time instant based on sorting and algorithm 1 (our technique with $M=100$ and $\epsilon = 0.01$).}
\label{fig11}
\end{figure}

In the following, we show how our technique in Algorithm 1 will lighten an impulsive noise which has a heavy-tailed distribution and is modeled as a mixture of two Gaussian distributions (with means equal to $0$ and standard deviations equal to $10^{-4}$ and $10$) as follows: 
\begin{align}
\label{impNoi}
\nu \sim 0.9\mathcal{N}(0,10^{-8}) + 0.1\mathcal{N}(0,100).  
\end{align}
The second term of above impulsive noise generates abnormally large noise samples or outliers. Scatter plot and histogram (with 100 bins with equal size) of 10000 samples drawn from this distribution are shown in Figure \ref{fig12}-a and Figure \ref{fig12}-b, respectively. Next, we deploy our running quartile technique in Algorithm 1 to obtain $Q_{1,n}$ and $Q_{3,n}$ and then calculate upper and lower extremes based on them. Now, we use these extremes to detect and exclude these major outliers from noise samples which results in scatter plot of non-outlier noise samples in Figure \ref{fig12}-c (974 noise samples out of total 10000 noise samples have been detected and excluded as major outliers). As seen in this Figure, non-outlier noise samples are concentrated around $\nu =0$. Finally, Figure \ref{fig12}-d shows histogram of these non-outlier noise samples that can be interpreted as histogram of samples drawn from lightened version of impulsive noise (\ref{impNoi}) (again number of equal-size bins is 100).

We can readily calculate the mean of random variable $\nu $ distributed according to heavy-tailed distribution (\ref{impNoi}) which is $0$. Figure \ref{fig13} depicts how sample mean approximation will converge to actual mean, i.e., $0$. As illustrated in this Figure, if we use all noise samples including major outliers convergence occurs with many fluctuations around the actual mean while when we exclude these major outliers, convergence to the mean of the lightened noise distribution happens very fast and more consistently and smoothly.

\begin{figure*}[t]
  \centering
\subfloat{\textbf{a.}}{
	\includegraphics[width=.4\linewidth]{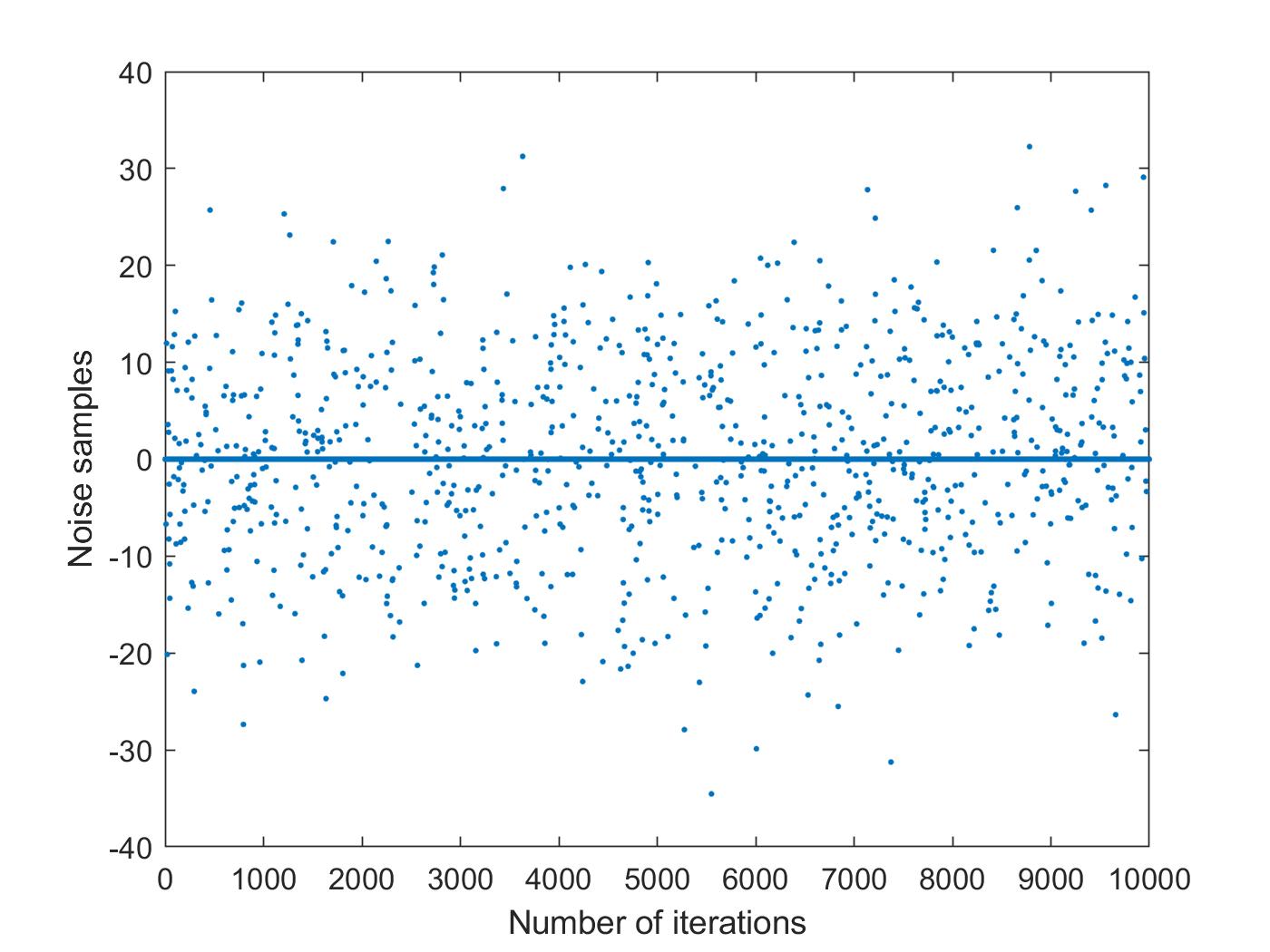} } 
\subfloat{\textbf{b.}}{
   \includegraphics[width=.4\linewidth]{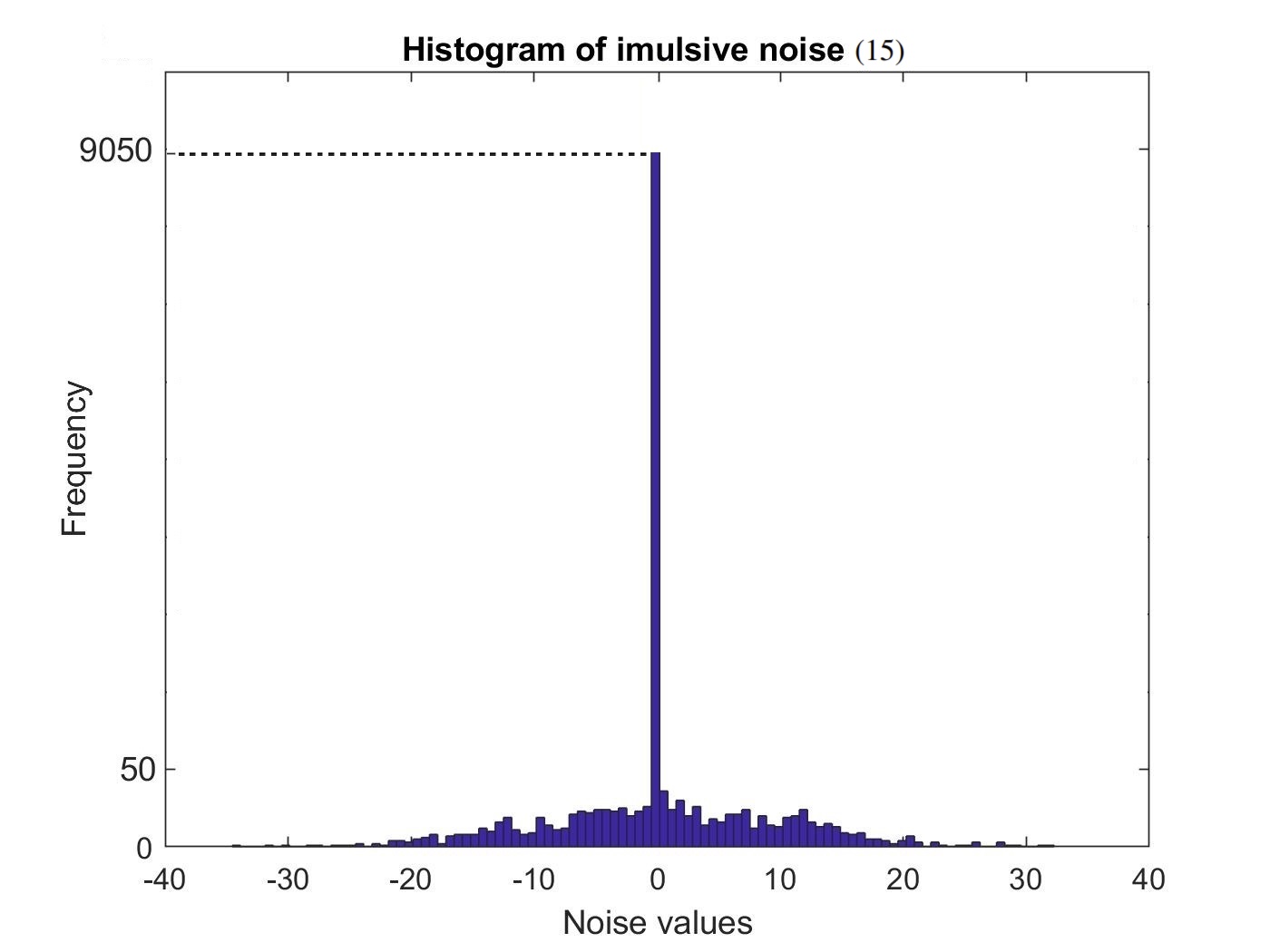} }
   
\subfloat{\textbf{c.}}{
   \includegraphics[width=.4\linewidth]{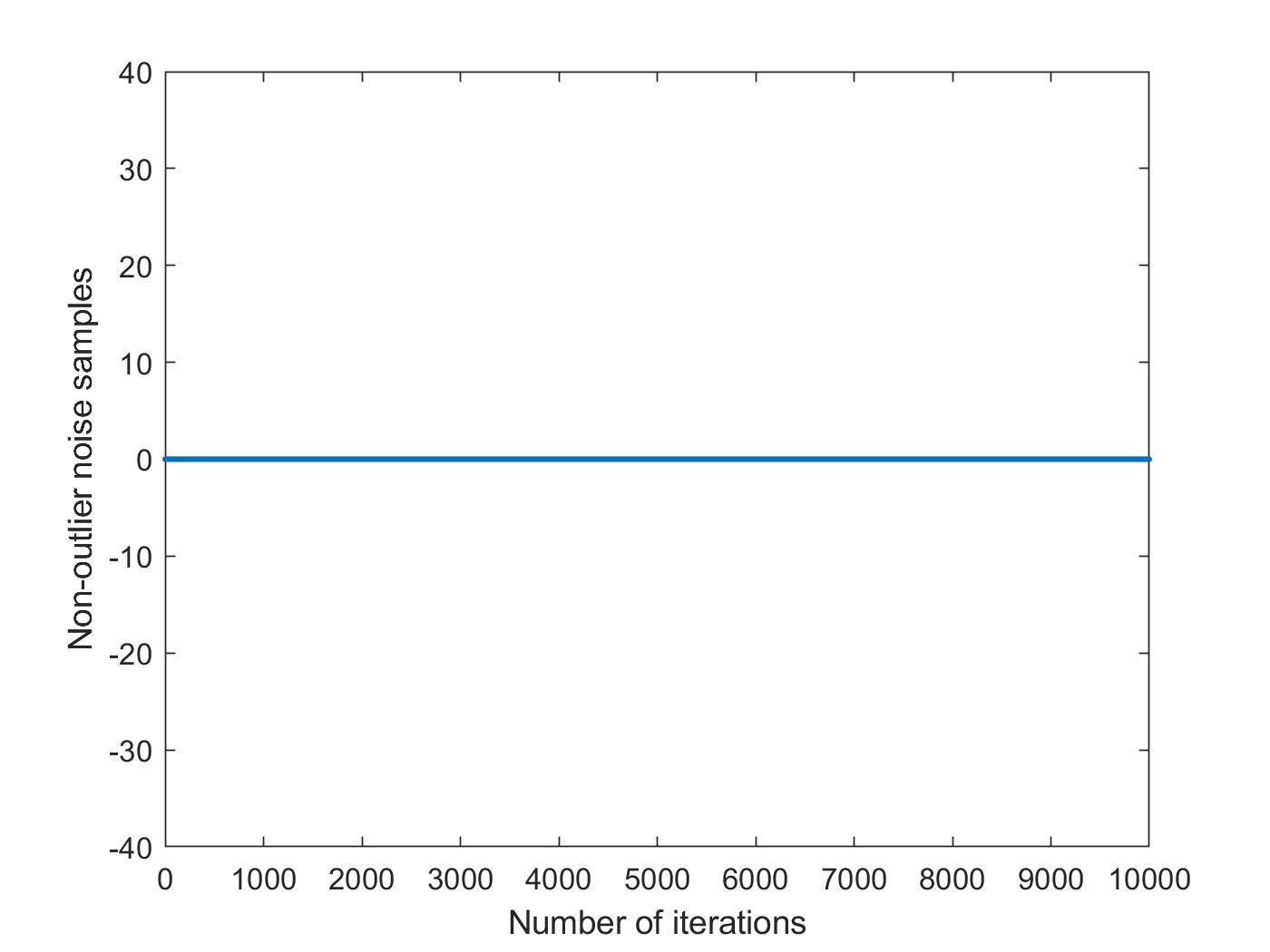} }
 \subfloat{\textbf{d.}}{
    \includegraphics[width=.4\linewidth]{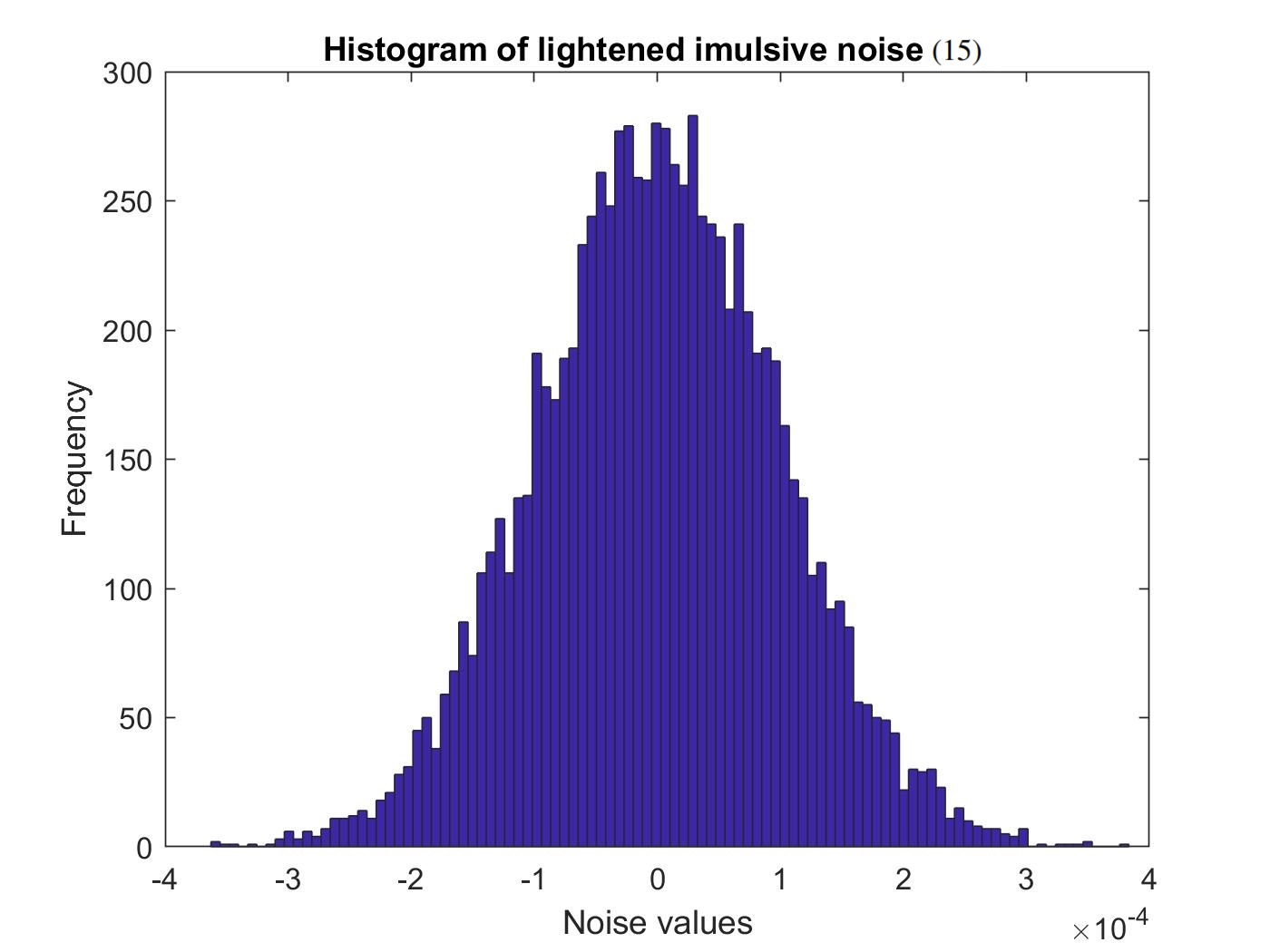} }
\caption{Scatter plot and histogram of noise samples drawn from impulsive noise (\ref{impNoi}) (Figures a. and b.) and those of noise samples drawn from lightened version of impulsive noise (\ref{impNoi}) obtained by using Algorithm 1 (Figures c. and d.).}
\label{fig12}
\end{figure*}
\begin{figure}[t!]
\centering
\includegraphics[scale=.16]{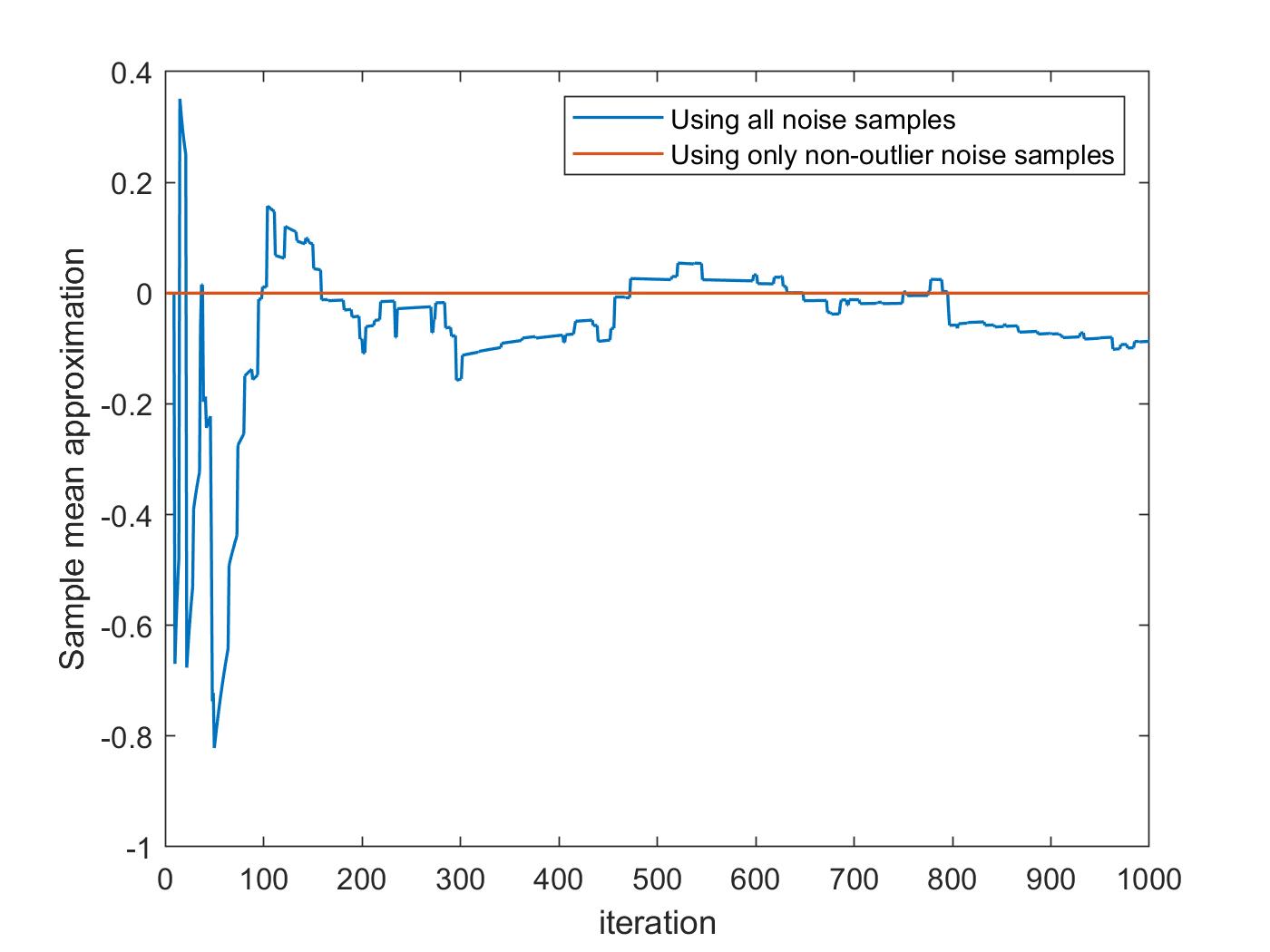}
\caption{Sample mean approximation using all noise (\ref{impNoi}) samples (including major outliers) versus using only non-outlier noise samples (detected based on Algorithm 1 and concept of the upper and lower extremes).}
\label{fig13}
\end{figure}
\subsection{Trimmed MEE}
In this subsection,  MEE and MEEF are compared with our proposed Trimmed MEE in Algorithm 2 for online linear regression in which we deploy our proposed error samples running quartile estimation technique (Algorithm 1) in order to detect and exclude major outliers (or abnormally large error samples) from learning process. Recall that for MEE and Trimmed MEE algorithms we add a bias (error sample mean) obtained from all error samples and non-outlier error samples, respectively to the output of the system in order to locate the final error samples around the origin. First, we consider exponential noises $\nu _1\sim Exp(\sqrt{2000} )$ and $\nu _2\sim Exp(\sqrt{2\times 10^{5}} )$ then Gaussian noise $\nu _3\sim \mathcal{N}(0,10^{-3}) $ which are not heavy-tailed, hence we do not need to be concerned about destructive effect of major outliers. In other words we expect to obtain similar results regardless of using MEE or Trimmed MEE for environments affected by these noises as shown in Figures \ref{fig14}-a, \ref{fig14}-b and \ref{fig14}-c, where for exponential noises Trimmed MEE even shows a slightly better steady state misalignment performance compared to MEE which is not surprising. Learning curve of MEEF also shows similar behaviour to MEE and Trimmed MEE, however for exponential noise MEEF performance deteriorates as SNR is increased inasmuch as increase in SNR means increase in $\lambda $ which results in larger gap between MEE and MCC cost functions \cite{8.5}.  
\begin{figure*}
\centering
\subfloat{\textbf{a.}}{
   \includegraphics[width=.3\linewidth]{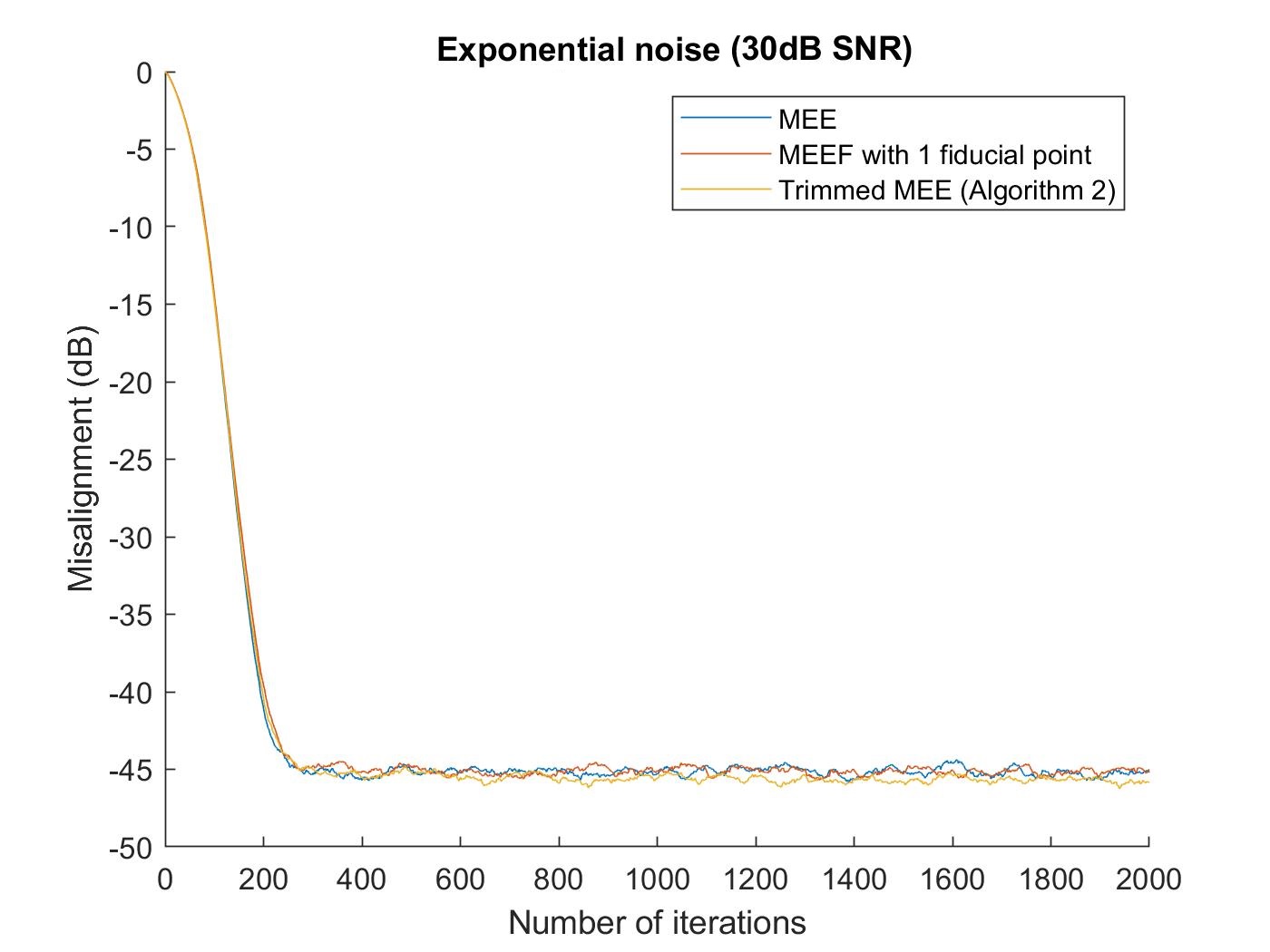} }
\subfloat{\textbf{b.}}{
  \includegraphics[width=.3\linewidth]{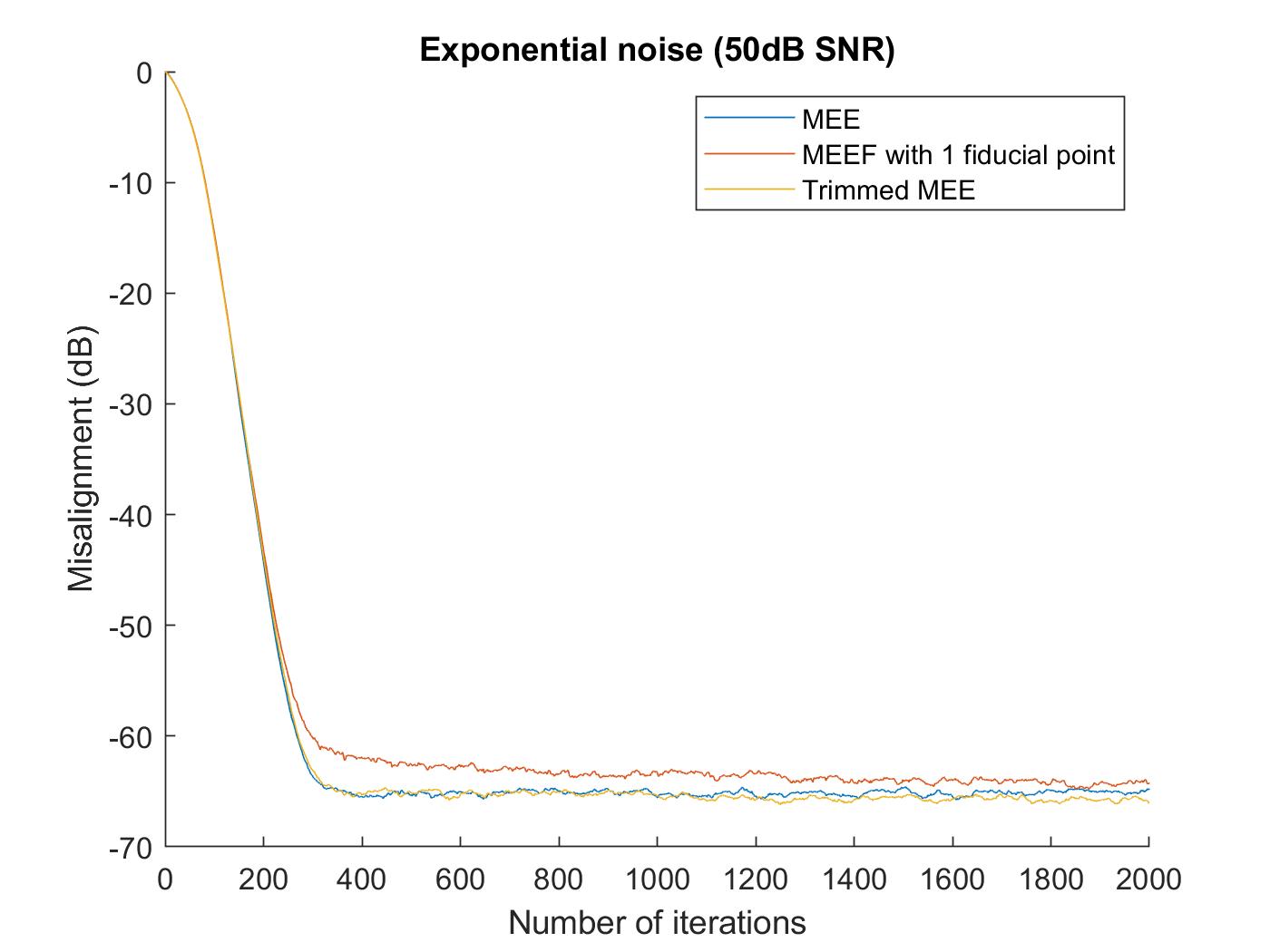} }
\subfloat{\textbf{c.}}{
  \includegraphics[width=.3\linewidth]{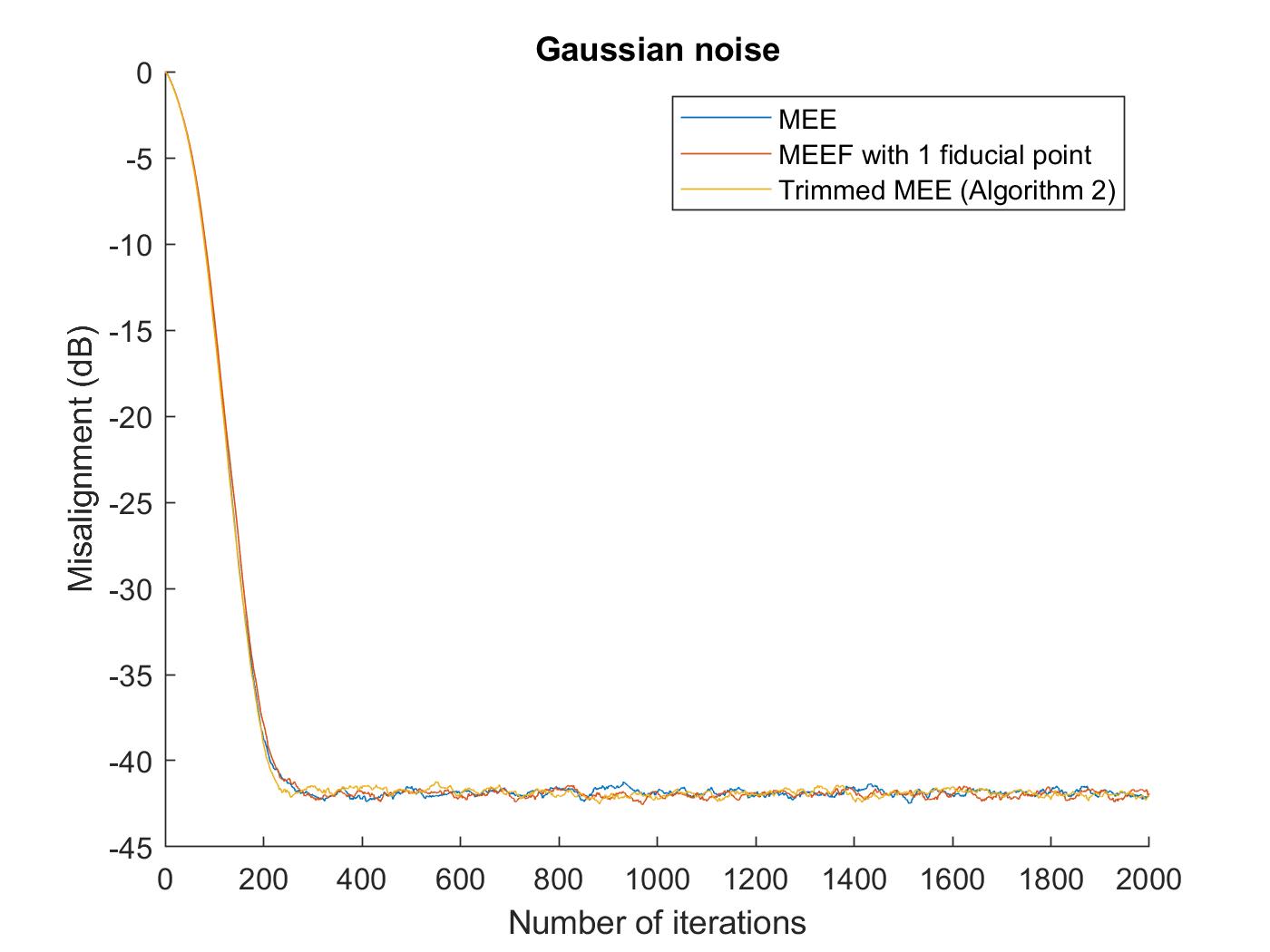} }

\subfloat{\textbf{d.}}{ 
  \includegraphics[width=.3\linewidth]{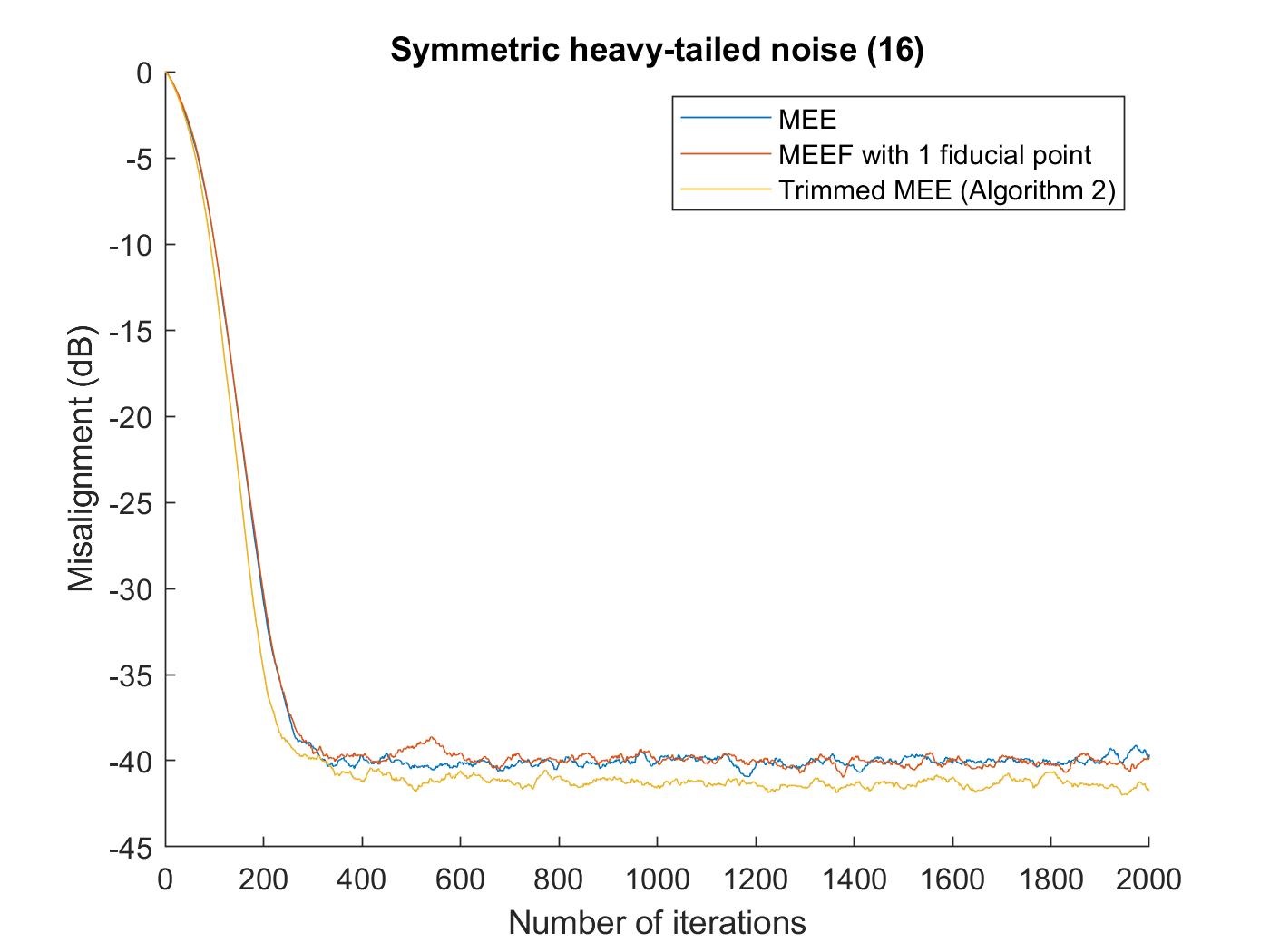} }
\subfloat{\textbf{e.}}{ 
  \includegraphics[width=.3\linewidth]{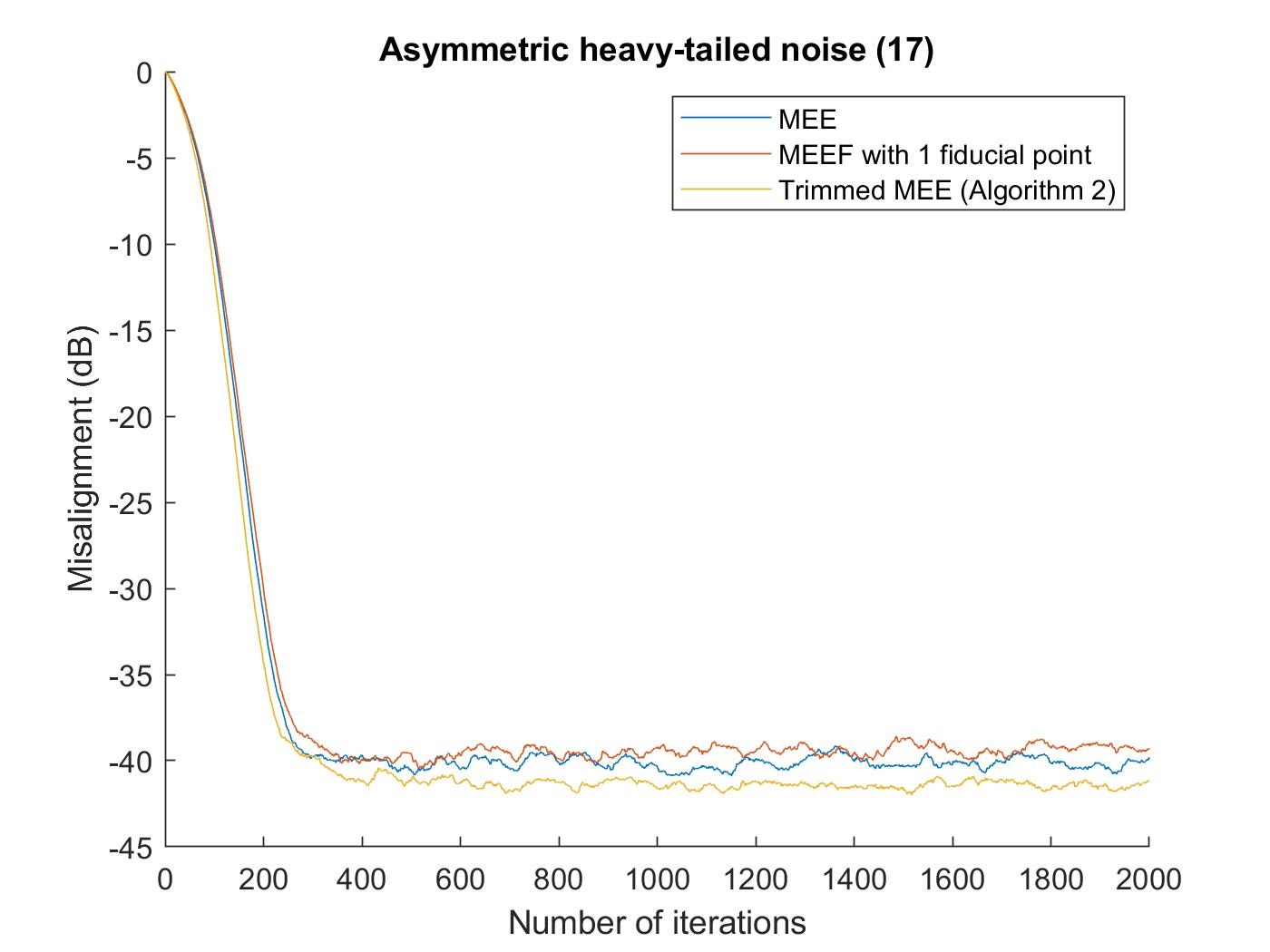} }
\subfloat{\textbf{f.}}{ 
  \includegraphics[width=.3\linewidth]{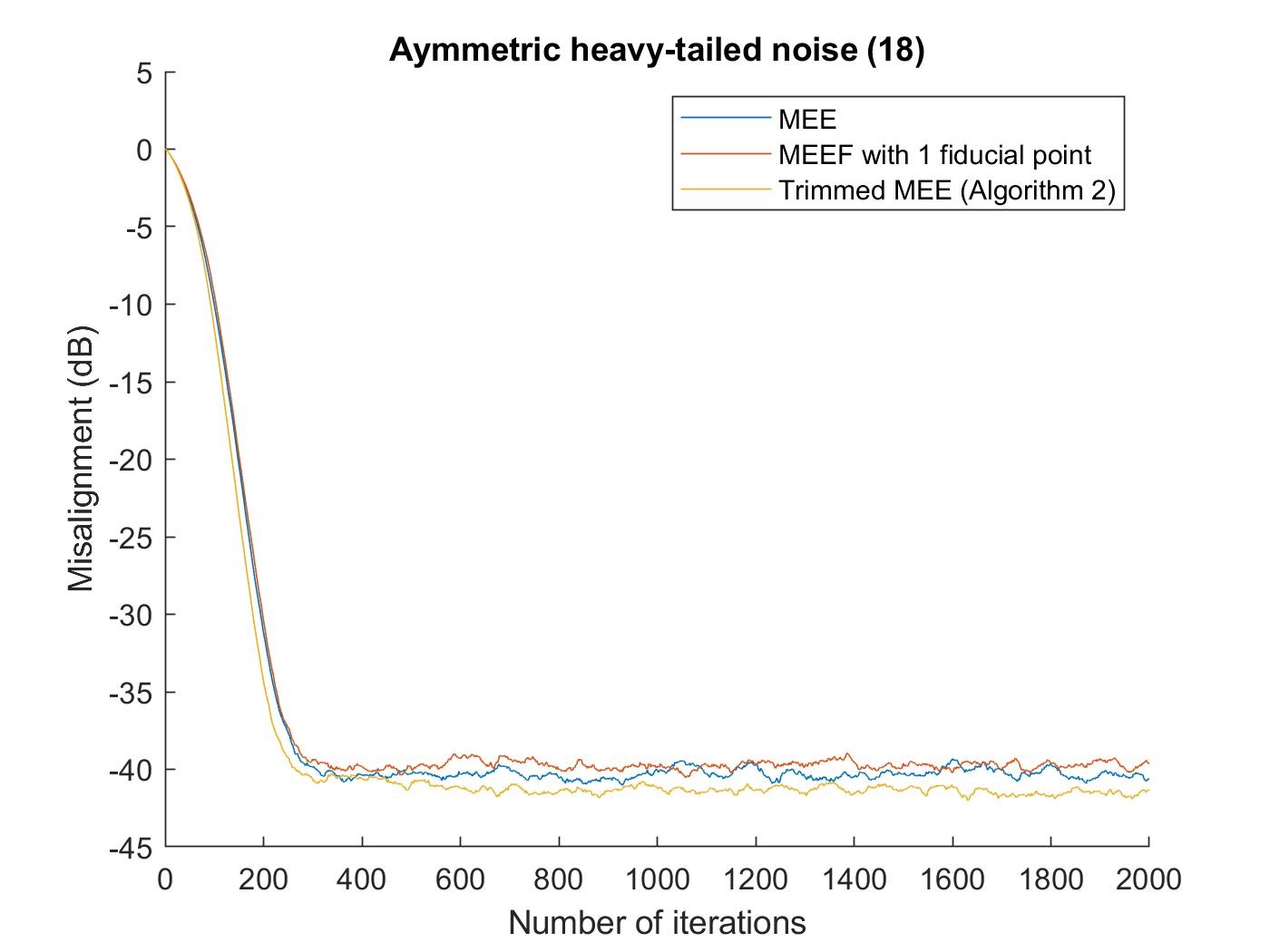} }
\caption{Learning curves of MEE, MEEF and Trimmed MEE under different noises averaged over 200 Monte Carlo simulations with $\mu = 0.05$ and $\sigma =1$.}
\label{fig14}
\end{figure*}

In order to show superiority of Trimmed MEE over MEE and MEEF in presence of heavy-tailed noises, we consider following symmetric and asymmetric mixture of Gaussians noises:
\begin{align}
\label{symNoi}
&\nu _1 \sim 0.9\mathcal{N}(0,10^{-3}) + 0.1\mathcal{N}(0,1000),  \\
\label{asymNoi}
&\nu _2 \sim 0.9\mathcal{N}(0,10^{-3}) + 0.1\mathcal{N}(10,1000).  \\
\label{asymNoi2}
&\nu _3 \sim 0.9\mathcal{N}(-5,10^{-3}) + 0.1\mathcal{N}(10,1000).  
\end{align}
As seen in Figures \ref{fig14}-d, \ref{fig14}-e and \ref{fig14}-f, for all symmetric and asymmetric heavy-tailed noises (\ref{symNoi}), (\ref{asymNoi}) and (\ref{asymNoi2}) Trimmed MEE outperforms MEE and MEEF from both convergence rate and steady state misalignment point of views. Learning curves of MEE and MEEF are similar, however for asymmetric noises (\ref{asymNoi}) and (\ref{asymNoi2}) MEE shows a slightly better performance. It is worth mentioning that all learning curves of Figure \ref{fig14} are obtained by averaging over 200 independent Monte Carlo simulations. The claimed strength of MEEF in literature is its ability to locate the majority of errors between system outputs and labels around the origin which becomes more clear when we do error analysis. However, this claim is valid as long as the gap between MEE and MCC cost functions is not large. This fact is shown in Table I where testing error analysis of these three algorithms (MEE, MEEF with 1 fiducial point and Trimmed MEE) is done for noises discussed in Figure \ref{fig14}. The error analysis is done based on the following metric called mean absolute error (MAE):
\begin{align}
\mathrm{MAE} =\frac{1}{T}\sum _{t=1}^T |d_t-y_t|=\frac{1}{T}\sum _{t=1}^T |e_t|,
\nonumber 
\end{align}
where $T$ is the number of testing samples. Again, we ran 200 independent Monte Carlo simulations, each with 2000 iterations ($\mu =0.05,~\sigma =1,~M=100$ and $\epsilon = 0.01$), learned the system parameters in each simulation and then obtained testing errors based on 2000 testing samples for each simulation, therefore each MAE shown in Table I is obtained by averaging over 200 MAEs. The smallest MAE amongst all three algorithms for each noise is highlighted in bold. As seen in this Table, for all noises Trimmed MEE gives the best result. More precisely, for exponential noise MEE and Trimmed MEE outperforms MEEF, as expected because for this noise gap between MEE and MCC cost functions is large. For Gaussian noise the performance of all algorithms is the same. For mixture of two Gaussians noises, which are heavy-tailed, as long as the gap between MEE and MCC is not large MEEF outperforms MEE, although Trimmed MEE still shows the best performance. However, once this noise makes a large gap between MEE and MCC cost functions, as noise (\ref{asymNoi2}) does, MEEF shows a weaker performance even than MEE while Trimmed MEE shows the best performance again.

\begin{table*}[!t]
\begin{center}
\caption{Testing Mean Absolute Errors of Different Algorithms under Different Noises}
\begin{tabular}{ c c c c }
\toprule
&     MEE     &     MEEF     &     Trimmed MEE     \\
\midrule
Exponential noise with 30dB SNR & 0.0171$\pm $0.0012 &0.0232$\pm $0.0016 &\textbf{0.0169}$\pm $\textbf{0.0010} \\
\midrule
Exponential noise with 50dB SNR & \textbf{0.0017}$\pm $\textbf{0.0001} & 0.0023$\pm $0.0002 & \textbf{0.0017}$\pm $\textbf{0.0001}\\
\midrule
Gaussian noise with 30dB SNR & 0.0262$\pm $0.0016 & 0.0262$\pm $0.0016 & 0.0262$\pm $0.0016 \\
\midrule
Mixture of Gaussians noise (16) & 2.6669$\pm $0.4625 &2.5824$\pm $0.4206 & \textbf{2.5485}$\pm $\textbf{0.4043}\\
\midrule
Mixture of Gaussians noise (17) & 3.4912$\pm $0.5754 &2.6762$\pm $0.4812 & \textbf{2.6469}$\pm $\textbf{0.4272}\\
\midrule
Mixture of Gaussians noise (18) & 4.0958$\pm $0.6527 &7.1222$\pm $0.4281 & \textbf{2.8526}$\pm $\textbf{0.5376} \\
\bottomrule
\end{tabular}
\end{center}
\end{table*}
\begin{figure*}[!t]
\centering
\subfloat{\textbf{a.}}{
   \includegraphics[width=.3\linewidth]{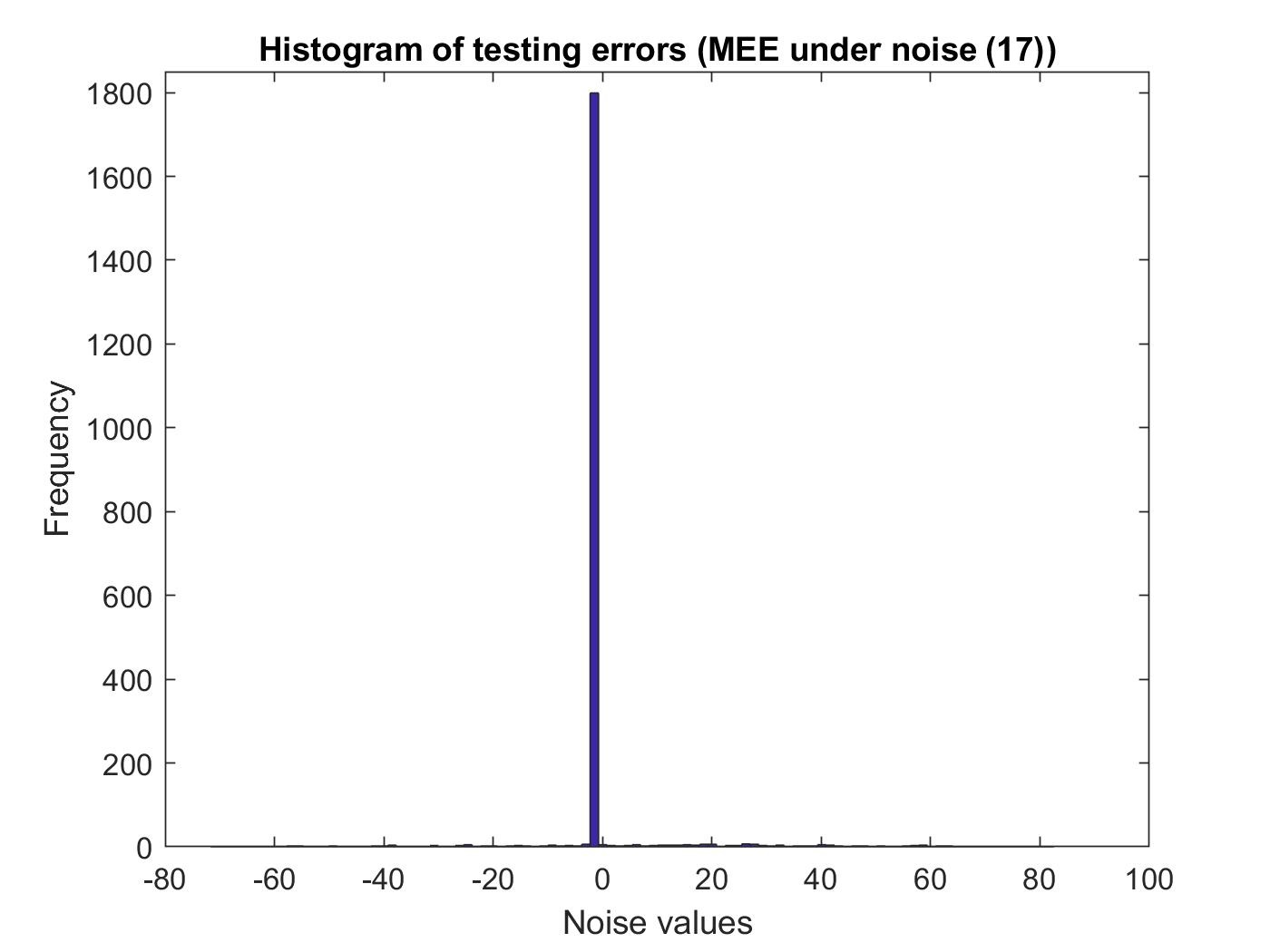} }
\subfloat{\textbf{b.}}{ 
   \includegraphics[width=.3\linewidth]{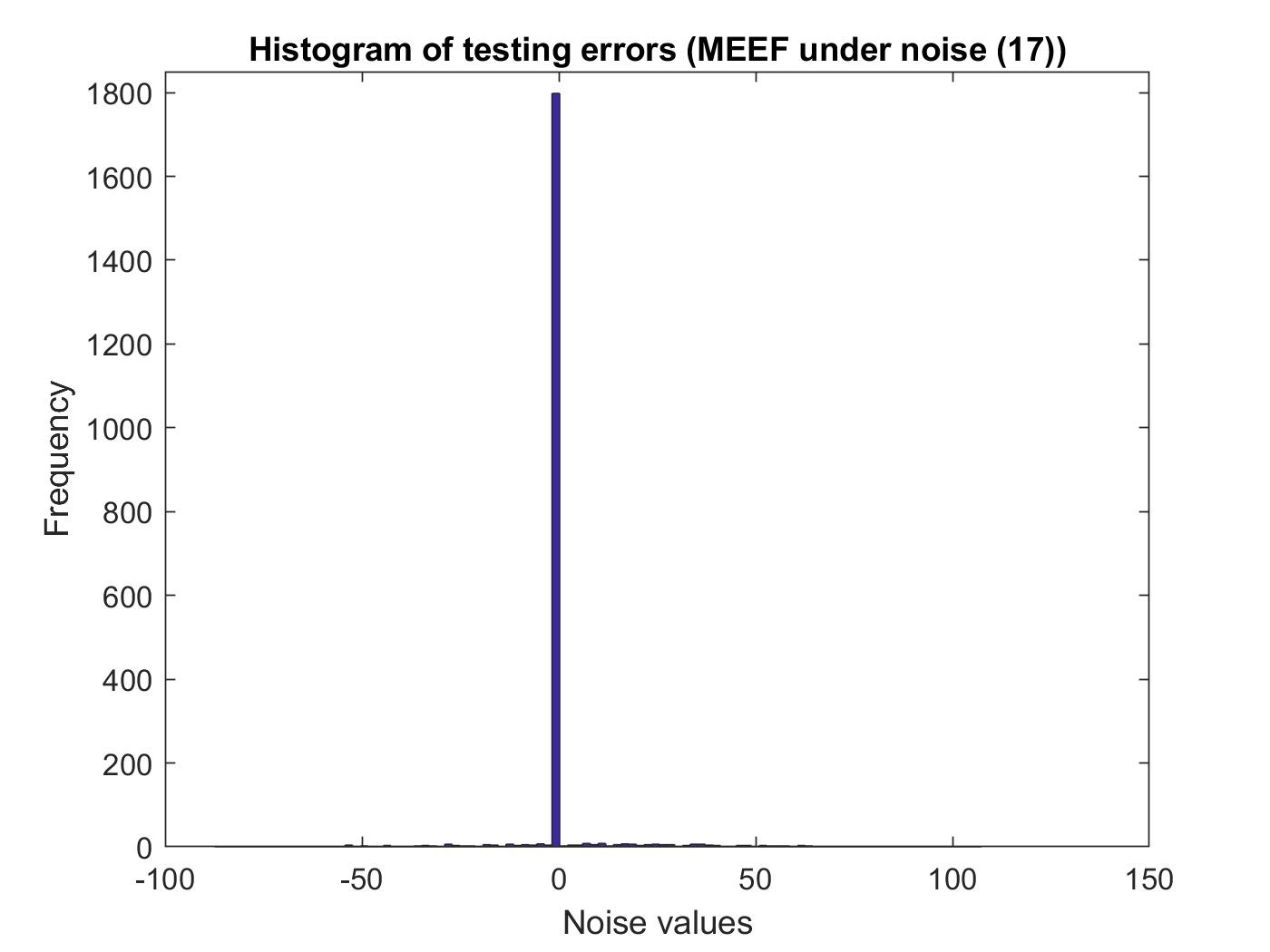} }
\subfloat{\textbf{c.}}{ 
  \includegraphics[width=.3\linewidth]{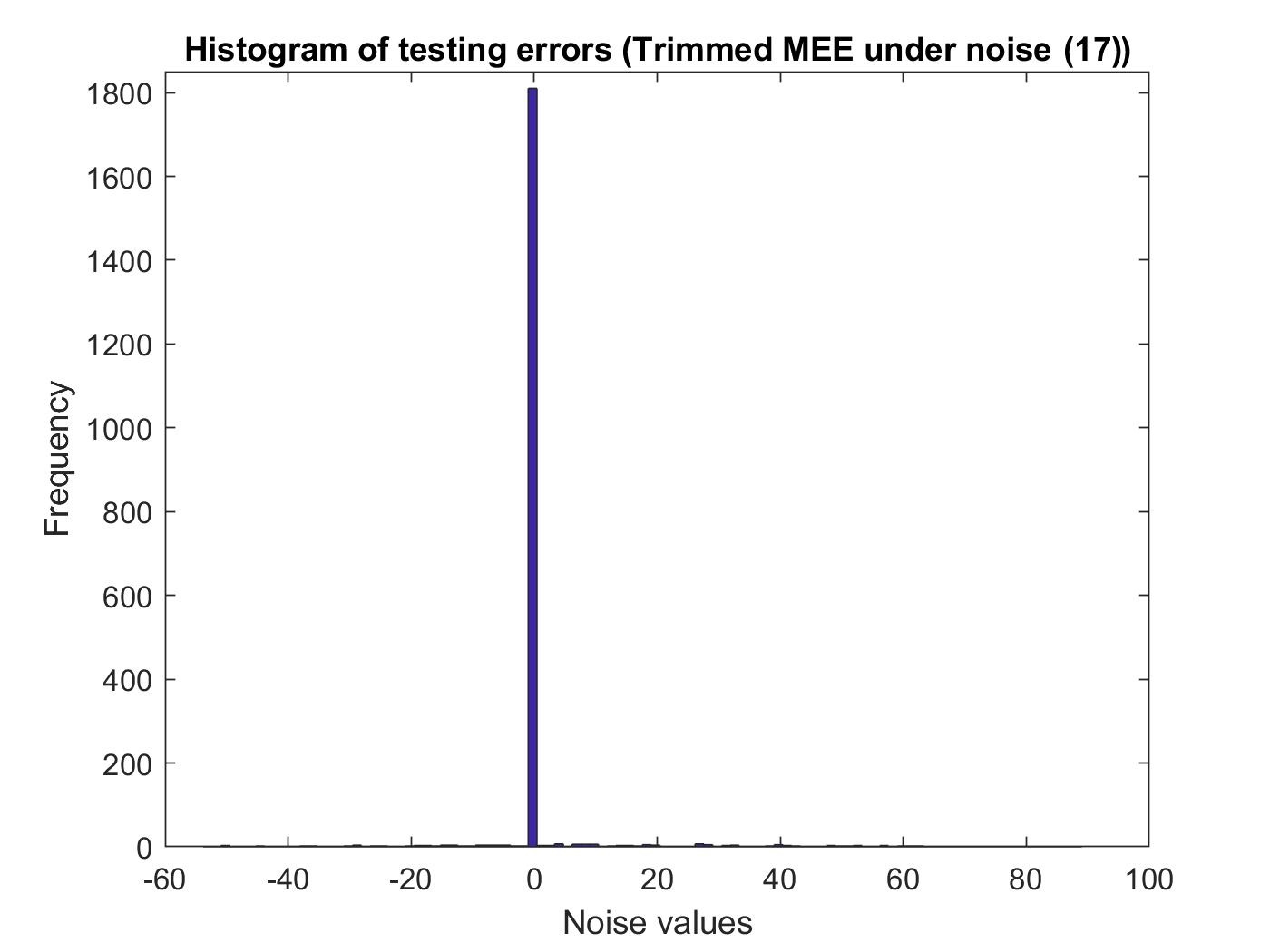} }

\subfloat{\textbf{d.}}{
  \includegraphics[width=.3\linewidth]{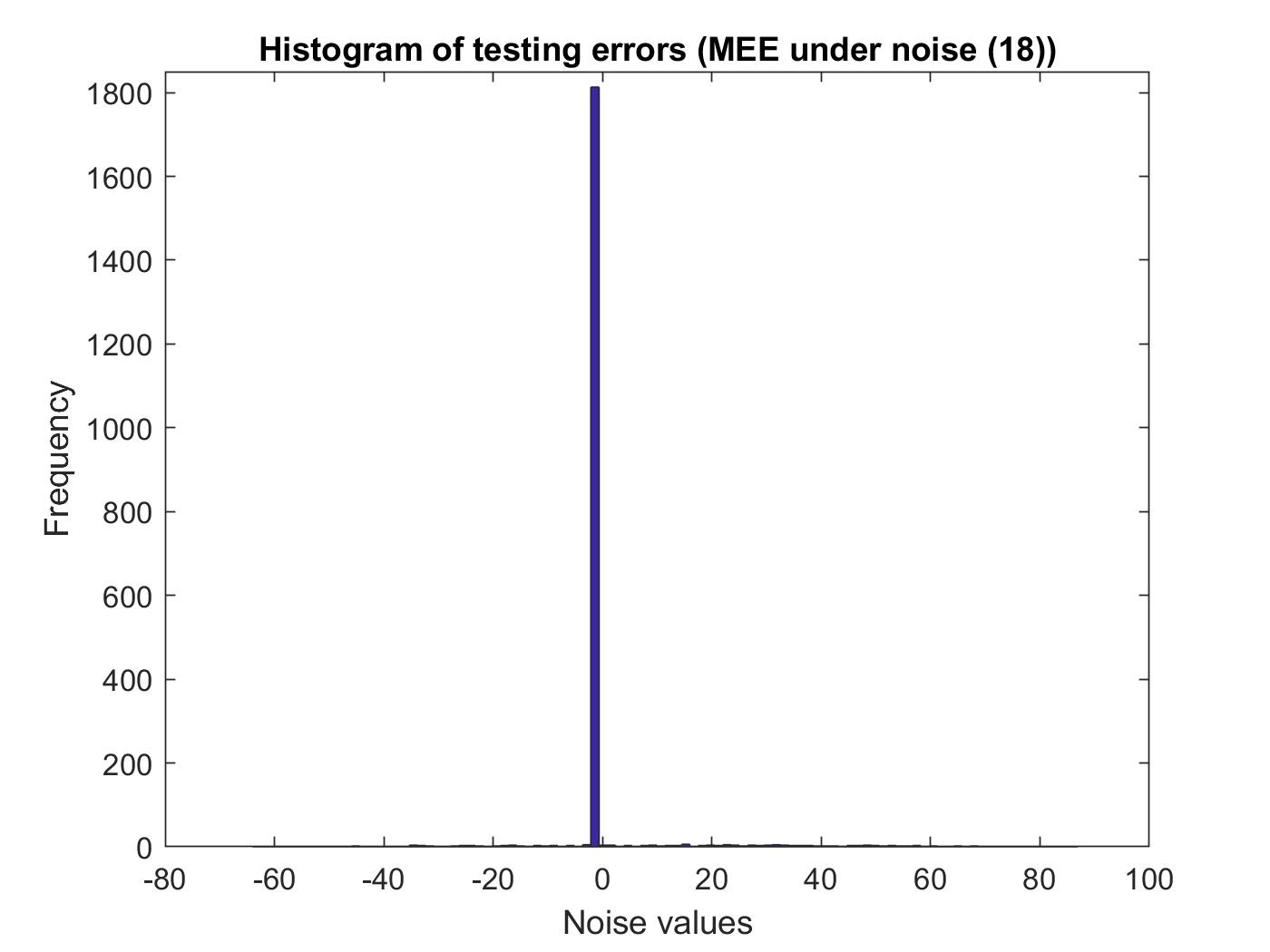} }
\subfloat{\textbf{e.}}{
  \includegraphics[width=.3\linewidth]{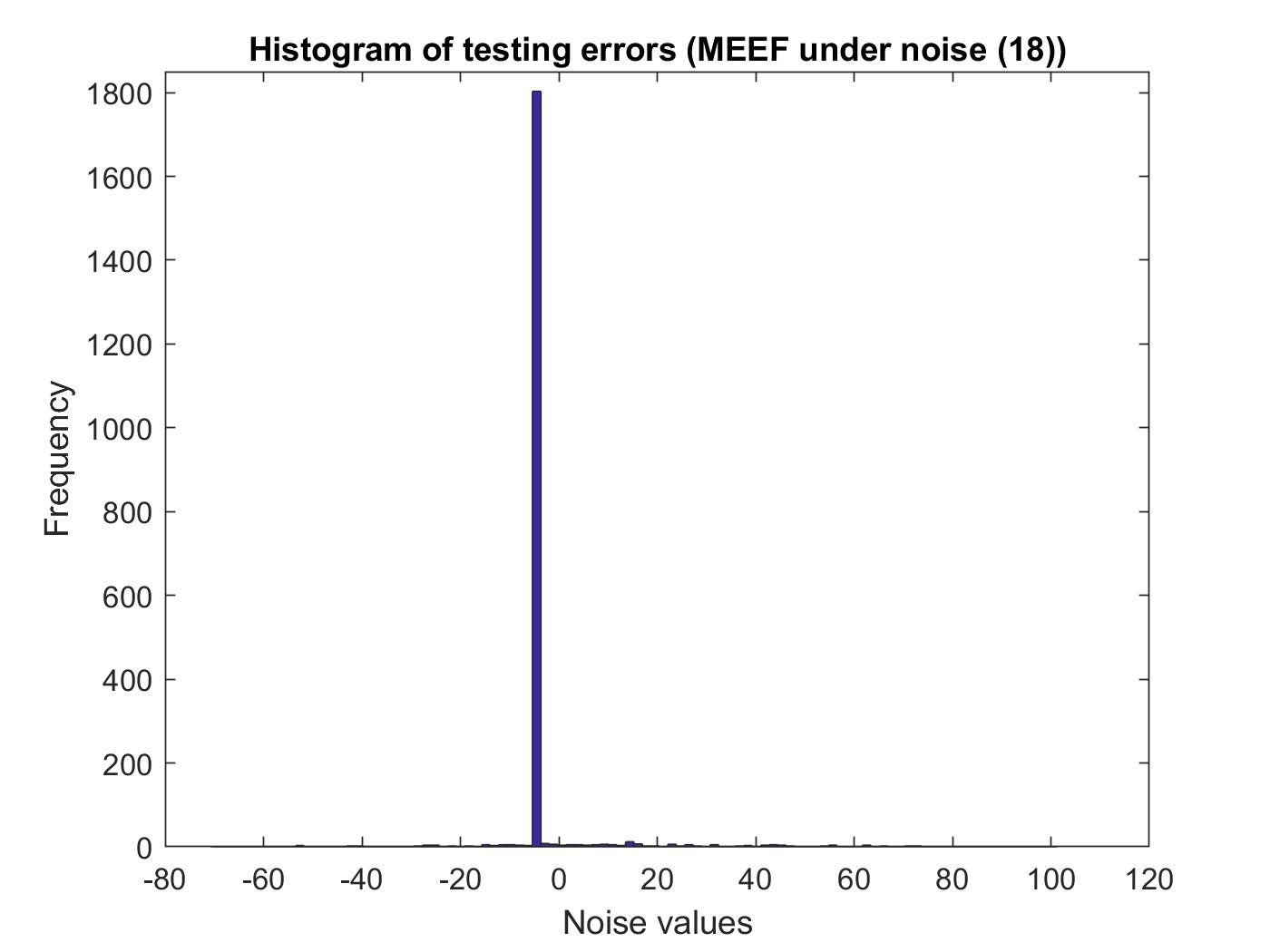} }
\subfloat{\textbf{f.}}{  \includegraphics[width=.3\linewidth]{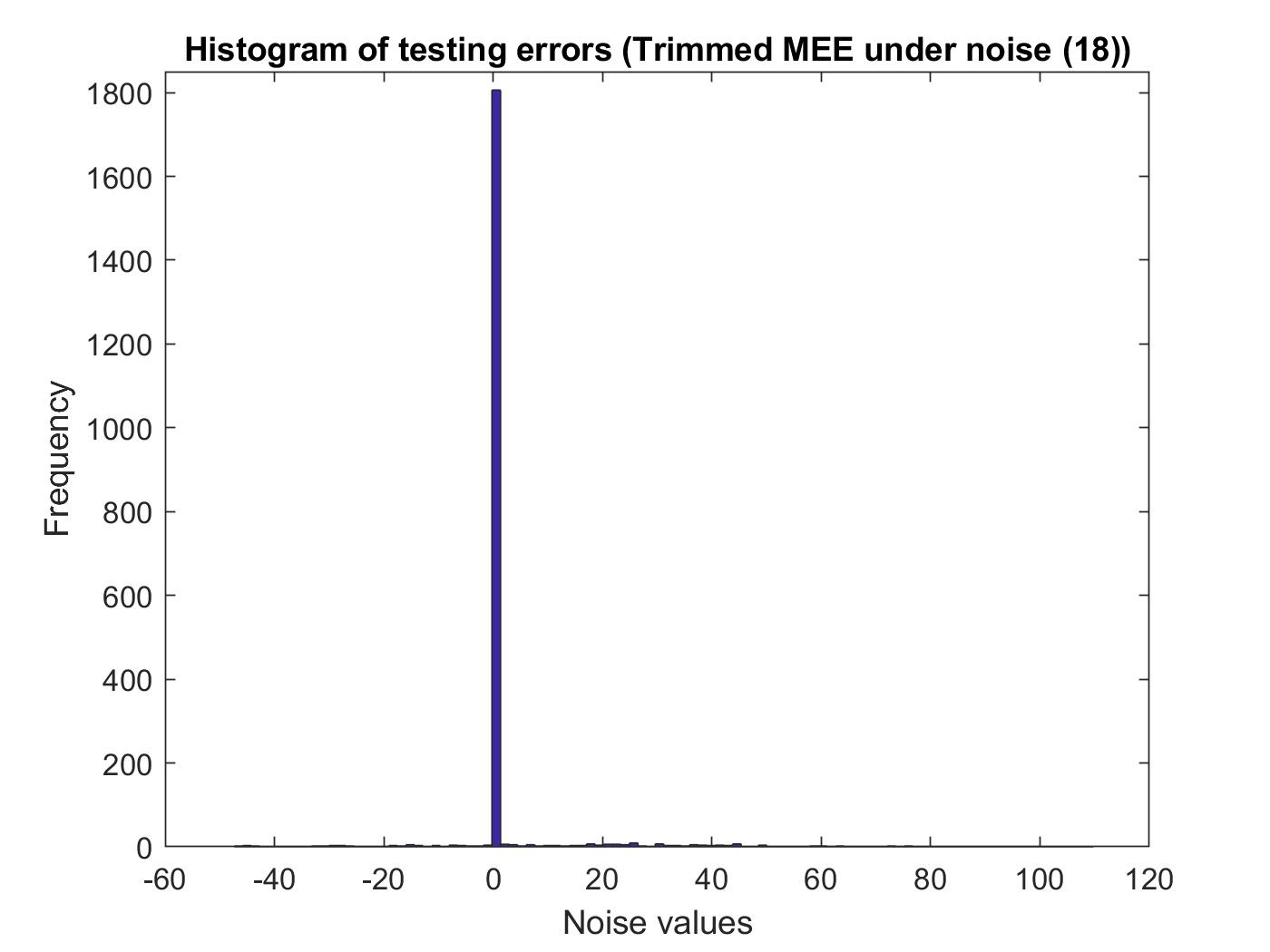} }
\caption{Testing error histograms of MEE, MEEF and Trimmed MEE obtained based on 2000 testing samples under Mixture of two Gaussians noises (\ref{asymNoi}) and (\ref{asymNoi2}).}
\label{fig15}
\end{figure*}
Figure \ref{fig15} shows the testing error histograms (obtained based on 2000 testing samples) for these three algorithms under noises (\ref{asymNoi}) and (\ref{asymNoi2}). As shown in this Figure, we can see again superiority of our algorithm where mass of the testing errors (approximately 90\% of them) is closer to the origin in Trimmed MEE compared to MEE and MEEF. Note that for noise (\ref{asymNoi2}), which makes large gap between MEE and MCC cost functions, MEEF even performs worse than MEE.

\section{Conclusion}
In this paper we address robust online linear regression in the presence of non-Gaussian noises (especially heavy-tailed noises). We can name many applications for online linear regression such as channel estimation and equalization, active noise cancellation, etc. Error entropy as a robust cost function has been utilized for robust learning under non-Gaussianity, however since entropy is shift-invariant we need to take some extra steps to locate the error PDF around the origin. To this end, two methods has been proposed nevertheless we show the shortcomings of these methods in this paper and propose our method. First, we propose an online algorithm in order to find the running quartiles of the error samples, then we use them to detect and eliminate major outliers from learning procedure based on MEE. We call our proposed method Trimmed MEE. Simulation results show the robustness of our algorithm to non-Gaussian (especially heavy-tailed) noises and its superiority over known methods in locating error PDF around the origin. In more details, proposed algorithm results in a learning curve with faster convergence to lower steady state misalignment and also achieves lower testing error compared to other algorithms. 

\printbibliography

@book{1,
  title={Information theoretic learning: Renyi's entropy and kernel perspectives},
  author={Principe, Jose C},
  year={2010},
  publisher={Springer Science \& Business Media}
}

@article{2,
  title={Underwater acoustic communication channels: Propagation models and statistical characterization},
  author={Stojanovic, Milica and Preisig, James},
  journal={IEEE communications magazine},
  volume={47},
  number={1},
  pages={84--89},
  year={2009},
  publisher={IEEE}
}

@article{3,
  title={Impulsive noise mitigation in underwater acoustic OFDM systems},
  author={Kuai, Xiaoyan and Sun, Haixin and Zhou, Shengli and Cheng, En},
  journal={IEEE Transactions on Vehicular Technology},
  volume={65},
  number={10},
  pages={8190--8202},
  year={2016},
  publisher={IEEE}
}

@article{4,
  title={Joint channel estimation and impulsive noise mitigation in underwater acoustic OFDM communication systems},
  author={Chen, Peng and Rong, Yue and Nordholm, Sven and He, Zhiqiang and Duncan, Alexander J},
  journal={IEEE Transactions on Wireless Communications},
  volume={16},
  number={9},
  pages={6165--6178},
  year={2017},
  publisher={IEEE}
}

@article{5,
  title={Joint channel and impulsive noise estimation in underwater acoustic OFDM systems},
  author={Chen, Peng and Rong, Yue and Nordholm, Sven and He, Zhiqiang},
  journal={IEEE Transactions on Vehicular Technology},
  volume={66},
  number={11},
  pages={10567--10571},
  year={2017},
  publisher={IEEE}
}

@article{6,
  title={Probability, random variables and stochastic processes},
  author={Papoulis, Athanasios and Saunders, H},
  year={2002}
}

@article{7,
  title={An error-entropy minimization algorithm for supervised training of nonlinear adaptive systems},
  author={Erdogmus, Deniz and Principe, Jose C},
  journal={IEEE Transactions on Signal Processing},
  volume={50},
  number={7},
  pages={1780--1786},
  year={2002},
  publisher={IEEE}
}

@article{8,
  title={Correntropy: Properties and applications in non-Gaussian signal processing},
  author={Liu, Weifeng and Pokharel, Puskal P and Principe, Jose C},
  journal={IEEE Transactions on signal processing},
  volume={55},
  number={11},
  pages={5286--5298},
  year={2007},
  publisher={IEEE}
}

@article{8.5,
  title={A new information theoretic relation between minimum error entropy and maximum correntropy},
  author={Heravi, Ahmad Reza and Hodtani, Ghosheh Abed},
  journal={IEEE Signal Processing Letters},
  volume={25},
  number={7},
  pages={921--925},
  year={2018},
  publisher={IEEE}
}

@article{10,
  title={Quantized minimum error entropy criterion},
  author={Chen, Badong and Xing, Lei and Zheng, Nanning and Principe, Jose C},
  journal={IEEE transactions on neural networks and learning systems},
  volume={30},
  number={5},
  pages={1370--1380},
  year={2018},
  publisher={IEEE}
}

@inproceedings{12,
  title={Estimating the information potential with the fast Gauss transform},
  author={Han, Seungju and Rao, Sudhir and Principe, Jose},
  booktitle={International Conference on Independent Component Analysis and Signal Separation},
  pages={82--89},
  year={2006},
  organization={Springer}
}

@article{12.1,
  title={Learning Theory Approach to Minimum Error Entropy Criterion.},
  author={Hu, Ting and Fan, Jun and Wu, Qiang and Zhou, Ding-Xuan},
  journal={Journal of Machine Learning Research},
  volume={14},
  number={2},
  year={2013}
}

@article{12.2,
  title={On optimal estimations with minimum error entropy criterion},
  author={Chen, Badong and Zhu, Yu and Hu, Jinchun and Zhang, Ming},
  journal={Journal of the Franklin Institute},
  volume={347},
  number={2},
  pages={545--558},
  year={2010},
  publisher={Elsevier}
}

@article{12.2.1,
  title={Insights into the robustness of minimum error entropy estimation},
  author={Chen, Badong and Xing, Lei and Xu, Bin and Zhao, Haiquan and Principe, Jose C},
  journal={IEEE transactions on neural networks and learning systems},
  volume={29},
  number={3},
  pages={731--737},
  year={2016},
  publisher={IEEE}
}

@article{13,
  title={Stochastic gradient algorithm under (h, $\varphi$)-entropy criterion},
  author={Chen, Badong and Hu, Jinchun and Pu, Li and Sun, Zengqi},
  journal={Circuits, Systems \& Signal Processing},
  volume={26},
  number={6},
  pages={941--960},
  year={2007},
  publisher={Springer}
}

@inproceedings{14,
  title={Using information theoretic learning techniques to train neural networks},
  author={Deb, Manas and Ogunfunmi, Tokunbo},
  booktitle={2017 51st Asilomar Conference on Signals, Systems, and Computers},
  pages={351--355},
  year={2017},
  organization={IEEE}
}

@article{15,
  title={Adaptive blind deconvolution of linear channels using Renyi's entropy with Parzen window estimation},
  author={Erdogmus, Deniz and Hild, Kenneth E and Principe, Jose C and Lazaro, Marcelino and Santamaria, Ignacio},
  journal={IEEE Transactions on Signal Processing},
  volume={52},
  number={6},
  pages={1489--1498},
  year={2004},
  publisher={IEEE}
}

@article{16,
  title={Minimum-entropy estimation in semi-parametric models},
  author={Wolsztynski, Eric and Thierry, Eric and Pronzato, Luc},
  journal={Signal Processing},
  volume={85},
  number={5},
  pages={937--949},
  year={2005},
  publisher={Elsevier}
}

@article{17,
  title={Local minima of information-theoretic criteria in blind source separation},
  author={Pham, Dinh-Tuan and Vrins, Frederic},
  journal={IEEE Signal Processing Letters},
  volume={12},
  number={11},
  pages={788--791},
  year={2005},
  publisher={IEEE}
}

@article{18,
  title={Entropy minimization for supervised digital communications channel equalization},
  author={Santamaria, Ignacio and Erdogmus, Deniz and Principe, Jose C},
  journal={IEEE Transactions on Signal Processing},
  volume={50},
  number={5},
  pages={1184--1192},
  year={2002},
  publisher={IEEE}
}

@article{18.1,
  title={Minimum error entropy criterion based channel estimation for massive-MIMO in VLC},
  author={Mitra, Rangeet and Bhatia, Vimal},
  journal={IEEE Transactions on Vehicular Technology},
  volume={68},
  number={1},
  pages={1014--1018},
  year={2018},
  publisher={IEEE}
}

@article{12.3,
  title={Some further results on the minimum error entropy estimation},
  author={Chen, Badong and Principe, Jose C},
  journal={Entropy},
  volume={14},
  number={5},
  pages={966--977},
  year={2012},
  publisher={Molecular Diversity Preservation International}
}

@article{12.4,
  title={Minimum error entropy Kalman filter},
  author={Chen, Badong and Dang, Lujuan and Gu, Yuantao and Zheng, Nanning and Principe, Jose  C},
  journal={IEEE Transactions on Systems, Man, and Cybernetics: Systems},
  year={2019},
  publisher={IEEE}
}

@book{19,
  title={Outlier detection: techniques and applications},
  author={SURI, NNR MURTY RANGA and Murty, M Narasimha and Athithan, G},
  year={2019},
  publisher={Springer}
}

@article{19.1,
  title={Adaptive filters},
  author={Haykin, Simon},
  journal={Signal Processing Magazine},
  volume={6},
  number={1},
  year={1999},
  publisher={Citeseer}
}

@book{20,
  title={Elements of information theory},
  author={Cover, Thomas M},
  year={1999},
  publisher={John Wiley \& Sons}
}

@article{20.1,
  title={The Fundamentals of Heavy Tails: Properties, Emergence, and Estimation},
  author={Nair, Jayakrishnan and Wierman, Adam and Zwart, Bert},
  journal={Preprint, California Institute of Technology},
  year={2020}
}

@article{20.2,
  title={On estimation of a probability density function and mode},
  author={Parzen, Emanuel},
  journal={The annals of mathematical statistics},
  volume={33},
  number={3},
  pages={1065--1076},
  year={1962},
  publisher={JSTOR}
}

@book{23,
  title={Nonparametric kernel density estimation and its computational aspects},
  author={Gramacki, Artur},
  year={2018},
  publisher={Springer}
}

@article{23.1,
  title={Generalized correntropy for robust adaptive filtering},
  author={Chen, Badong and Xing, Lei and Zhao, Haiquan and Zheng, Nanning and Principe, Jose C},
  journal={IEEE Transactions on Signal Processing},
  volume={64},
  number={13},
  pages={3376--3387},
  year={2016},
  publisher={IEEE}
}

@article{23.2,
  title={Mixture correntropy for robust learning},
  author={Chen, Badong and Wang, Xin and Lu, Na and Wang, Shiyuan and Cao, Jiuwen and Qin, Jing},
  journal={Pattern Recognition},
  volume={79},
  pages={318--327},
  year={2018},
  publisher={Elsevier}
}

@article{22.1,
  title={Large scale online learning},
  author={Bottou, Leon and LeCun, Yann},
  journal={Advances in neural information processing systems},
  volume={16},
  pages={217--224},
  year={2004}
}

@article{34NEW,
  title={Robust constrained adaptive filtering under minimum error entropy criterion},
  author={Peng, Siyuan and Ser, Wee and Chen, Badong and Sun, Lei and Lin, Zhiping},
  journal={IEEE Transactions on Circuits and Systems II: Express Briefs},
  volume={65},
  number={8},
  pages={1119--1123},
  year={2018},
  publisher={IEEE}
}

@inproceedings{23.4,
  title={Kernel width adaptation in information theoretic cost functions},
  author={Singh, Abhishek and Principe, Jose C},
  booktitle={2010 IEEE International Conference on Acoustics, Speech and Signal Processing},
  pages={2062--2065},
  year={2010},
  organization={IEEE}
}

@article{23.5,
  title={Recursive maximum correntropy learning algorithm with adaptive kernel size},
  author={Radmanesh, Hamid and Hajiabadi, Mojtaba},
  journal={IEEE Transactions on Circuits and Systems II: Express Briefs},
  volume={65},
  number={7},
  pages={958--962},
  year={2017},
  publisher={IEEE}
}

@inproceedings{23.6,
  title={An adaptive kernel width update for correntropy},
  author={Zhao, Songlin and Chen, Badong and Principe, Jose C},
  booktitle={The 2012 international joint conference on neural networks (ijcnn)},
  pages={1--5},
  year={2012},
  organization={IEEE}
}

@inproceedings{23.7,
  title={A switch kernel width method of correntropy for channel estimation},
  author={Wang, Weihua and Zhao, Jihong and Qu, Hua and Chen, Badong and Principe, Jose C},
  booktitle={2015 International Joint Conference on Neural Networks (IJCNN)},
  pages={1--7},
  year={2015},
  organization={IEEE}
}

@inproceedings{23.92,
  title={Mitigating outlier effect in online regression: An efficient usage of error correntropy criterion},
  author={Bahrami, Sajjad and Tuncel, Ertem},
  booktitle={2020 International Joint Conference on Neural Networks (IJCNN)},
  pages={1--6},
  year={2020},
  organization={IEEE}
}

@article{23.8,
  title={Adaptive filtering under a variable kernel width maximum correntropy criterion},
  author={Huang, Fuyi and Zhang, Jiashu and Zhang, Sheng},
  journal={IEEE Transactions on Circuits and Systems II: Express Briefs},
  volume={64},
  number={10},
  pages={1247--1251},
  year={2017},
  publisher={IEEE}
}

@article{23.9,
  title={An improved variable kernel width for maximum correntropy criterion algorithm},
  author={Shi, Long and Zhao, Haiquan and Zakharov, Yuriy},
  journal={IEEE Transactions on Circuits and Systems II: Express Briefs},
  volume={67},
  number={7},
  pages={1339--1343},
  year={2018},
  publisher={IEEE}
}

@inproceedings{23.91,
  title={A new approach to online regression based on maximum correntropy criterion},
  author={Bahrami, Sajjad and Tuncel, Ertem},
  booktitle={2019 IEEE 29th International Workshop on Machine Learning for Signal Processing (MLSP)},
  pages={1--6},
  year={2019},
  organization={IEEE}
}

@book{21,
  title={A Modern Introduction to Probability and Statistics: Understanding why and how},
  author={Dekking, Frederik Michel and Kraaikamp, Cornelis and Lopuhaa, Hendrik Paul and Meester, Ludolf Erwin},
  year={2005},
  publisher={Springer Science \& Business Media}
}

@inproceedings{22,
  title={Error entropy, correntropy and m-estimation},
  author={Liu, Weifeng and Pokharel, PP and Principe, JC},
  booktitle={2006 16th IEEE Signal Processing Society Workshop on Machine Learning for Signal Processing},
  pages={179--184},
  year={2006},
  organization={IEEE}
}

@article{23.93,
  title={Fixed-point minimum error entropy with fiducial points},
  author={Xie, Yuqing and Li, Yingsong and Gu, Yuantao and Cao, Jiuwen and Chen, Badong},
  journal={IEEE Transactions on Signal Processing},
  volume={68},
  pages={3824--3833},
  year={2020},
  publisher={IEEE}
}

@book{191,
  title={Exploratory data analysis},
  author={Tukey, John W and others},
  volume={2},
  year={1977},
  publisher={Reading, Mass.}
}

@book{nuquan,
  title={Introduction to data compression},
  author={Sayood, Khalid},
  year={2017},
  publisher={Morgan Kaufmann}
}

\end{document}